%% file: mainrtasv2.tex
\pdfoutput=1

\documentclass[11pt,a4paper]{article}
\usepackage[T1]{fontenc}
\usepackage[utf8]{inputenc}

\usepackage{amsmath,amssymb,amsthm}

\usepackage{graphicx}
\usepackage{mathtools}
\usepackage{geometry}
\usepackage{caption}
\usepackage{subcaption}
\usepackage{float}
\usepackage[colorlinks=true,linkcolor=blue,citecolor=blue,urlcolor=blue]{hyperref}  
\usepackage{bookmark}
\usepackage[capitalise,nameinlink,noabbrev]{cleveref}
\newcommand{\R}{\mathbb{R}}
\usepackage{microtype}
\usepackage{booktabs,tabularx}
\usepackage[most]{tcolorbox}
\usepackage{enumitem}

\newcommand{\Group}{\mathcal{G}}

\usepackage[numbers,sort&compress]{natbib}

\newcolumntype{Y}{>{\raggedright\arraybackslash}X} 

\DeclareMathOperator{\pr}{pr}

\newtheorem{assumption}{Assumption}
\newtheorem{theorem}{Theorem}
\newtheorem{proposition}{Proposition}
\newtheorem{definition}{Definition}
\newtheorem{lemma}{Lemma}
\newtheorem{remark}{Remark}
\newtheorem{corollary}{Corollary}

\newcommand{\ip}[2]{\langle #1,#2\rangle}

\newcommand{\transpose}{\mathsf{T}}

\title{Tangential Action Spaces: Geometry, Memory and Cost in Holonomic and Nonholonomic Agents}
\author{Marcel Blattner\\
\small Applied AI Research Lab\\
\small Lucerne University of Applied Sciences and Arts\\
\small \texttt{marcel.blattner@hslu.ch}}

\date{}
\begin{document}
\maketitle
\tableofcontents
\newpage

\begin{abstract}
Living systems often negotiate a trade-off between energetic efficiency and the ability to retain path‑dependent effects. We introduce Tangential Action Spaces (TAS), a geometric formalism that represents embodied agents as hierarchies of manifolds linked by projection maps from physical states to cognitive representations and onward to intentional goals. Lifting intentions back to actions can proceed along multiple routes, which generally differ in energy cost and in whether they leave memory-like traces.

Under the stated assumptions, we establish three results. (i) When the physical‑to‑cognitive projection is locally invertible, there is a unique lift that minimises instantaneous energy and produces no path‑dependent memory; any lift that induces memory entails strictly positive excess energy. (ii) When multiple physical states map to the same cognitive state (fibration), the energy‑minimising lift is the weighted pseudoinverse determined by the physical metric. (iii) In systems that accumulate holonomy, excess energy grows quadratically with the size of the induced memory (relative to the metric‑lift baseline) for sufficiently small loops, providing a local cost–memory relationship.

These statements motivate an organising classification of embodied systems by how path dependence can arise (intrinsically conservative, conditionally conservative, geometrically nonconservative, dynamically nonconservative). We illustrate the framework with numerical examples representative of each case, finding behaviour consistent with the analysis. We further consider a reflective extension (rTAS) in which the perceptual map depends on a learnable model state; a block metric formalises an effort–learning trade‑off, and cross‑curvature terms couple physical and model holonomy. Simulations of single‑ and two‑agent settings exhibit behaviours such as role asymmetries and sensitivity to coupling.

Overall, TAS offers a geometric language for relating embodiment, memory and energetic cost. The framework suggests testable predictions and design considerations for biological and robotic systems, while clarifying the conditions under which path dependence is expected and what it may cost.
\end{abstract}


\section{Introduction}\label{sec:intro}

Formalizing the coupling between physical embodiment and cognitive processes remains central to understanding life‑like agency \cite{varela1991embodied}. While various approaches have tackled this challenge \cite{varela1991embodied,beer1995dynamical,friston2010free}, a unified mathematical framework that captures both the geometric and energetic aspects of embodied cognition has remained elusive. Tangential Action Spaces (TAS) address this gap through differential geometry, modeling agents as hierarchical smooth manifolds connected by projection maps.

The critical operation in TAS is lifting cognitive changes to physical actions, and it depends fundamentally on the geometric structure of these projections. This paper establishes that the rank properties of the physical $(P)$ to cognitive $(C)$ projection $\Phi$ dictate both the mathematical framework for lift operations and the mechanisms by which path‑dependent memory emerges. We reveal that systems naturally divide into two classes: those with locally bijective projections (diffeomorphisms) that require prescribed dynamics for path dependence, and those with genuine fibre structures (fibrations) that support intrinsic geometric holonomy through connection curvature.
See \Cref{fig:holonomy} for a visual overview of how closed loops in $C$ can lift to open paths in $P$ (the holonomy gap).

Perhaps most significantly, we demonstrate that path‑dependent memory incurs an energetic cost, revealing a fundamental trade‑off between behavioral efficiency and memory capacity. This cost–memory duality provides a principled explanation for the diversity of embodied strategies observed in both biological and artificial systems.

\paragraph{Biological Grounding.}
In biological systems, the physical manifold $P$ encompasses not merely spatial coordinates but the full repertoire of bodily states, muscle activation patterns, proprioceptive configurations, metabolic conditions and neural dynamics. The projection $\Phi: P \to C$ implements a fundamental coarse‑graining operation, where myriad physical configurations (different muscle tensions, joint angles or neural firing patterns) map to the same cognitive representation. This dimension reduction reflects how organisms extract task‑relevant information from their high‑dimensional physical states. Similarly, the projection $\Psi: C \to I$ further abstracts cognitive states into intentional goals, where multiple ways of thinking about a task collapse into a single objective. The lift operations then implement the inverse process: transforming abstract intentions back through cognitive plans into concrete physical actions, necessarily choosing specific instantiations from the many possibilities.

\medskip

\subsection*{Notation and Standing Conventions}
\addcontentsline{toc}{subsection}{Notation and Standing Conventions}

\begin{table}[H]
  \centering
  \caption{Notation used throughout the manuscript.}
  \label{tab:notation}
  \footnotesize
  \begin{tabularx}{\textwidth}{@{} l l l Y @{}}
    \toprule
    \textbf{Symbol} & \textbf{Type} & \textbf{Where} & \textbf{Meaning / Role} \\
    \midrule
    $(P,G)$ & Riem. manifold, $\dim P=m$ & \Cref{def:TAS} & Physical manifold with metric $G$. \\
    $C$ & Manifold, $\dim C=n$ & \Cref{def:TAS} & Cognitive manifold. \\
    $I$ & Manifold, $\dim I=k$ & \Cref{def:TAS} & Intentional manifold. \\
    $\Phi:P\!\to\!C$ & Surjective submersion & \Cref{def:TAS} & Perception / coarse–graining map. \\
    $\Psi:C\!\to\!I$ & Surjective submersion & \Cref{def:TAS} & Cognitive–intentional map. \\
    $D\Phi_p$ & $T_pP\!\to\!T_{\Phi(p)}C$ & \Cref{def:lift} & Differential of $\Phi$ at $p$. \\
    $V_pP=\ker D\Phi_p$ & Subspace of $T_pP$ & \Cref{def:ehresmann} & Vertical bundle (fibre directions). \\
    $H_p$ & Subspace of $T_pP$ & \Cref{def:ehresmann} & Horizontal subspace (chosen connection). \\
    $\mathrm{pr}^H_p,\mathrm{pr}^V_p$ & Projections & §\ref{subsubsec:lift-fibration} & $G$‑orthogonal projections to $H_p$, $V_pP$. \\
    $\mathcal{L}$ & Lift operation & \Cref{def:lift} & $\mathcal{L}:\Phi^{*}TC\!\to\!TP$ with $D\Phi\!\circ\!L=\mathrm{id}$. \\
    $\mathcal{L}_{\text{geom}}$ & $D\Phi^{-1}$ & \Cref{prop:unique_lift} & Unique lift in diffeomorphic case. \\
    $\mathcal{L}_{\text{metric}}$ & $G^{-1}\!D\Phi^\top(D\Phi G^{-1}D\Phi^\top)^{-1}$ & \Cref{prop:metriclift} & Energy-minimising lift (fibration). \\
    $\omega$ & Connection 1‑form & §\ref{subsec:geom-hol}, \Cref{def:curvhol} & Encodes horizontality; Abelian case $F=d\omega$. \\
    $F$ & Curvature 2‑form & §\ref{subsec:geom-hol}, \Cref{def:curvhol} & $F=d\omega+\omega\!\wedge\!\omega$; Abelian: $F=d\omega$. \\
    $\mathcal{E}$ & Scalar & §\ref{sec:cost} & Travel cost $\int\!\|\dot u\|_G^2\,dt$. \\
    $g_C$ & Riem. metric on $C$ & §\ref{sec:intentional} & Used for gradients on $C$. \\
    $g_I$ & Riem. metric on $I$ & §\ref{sec:intentional} & Used for gradients on $I$. \\

    \bottomrule
  \end{tabularx}
\end{table}
\begin{assumption}[Standing Assumptions]\label{ass:standing}
Throughout this paper, we assume:
\begin{enumerate}[label=\textup{(A\arabic*)}]
  \item fibre bundles have Abelian structure groups.
  \item All structures are time-invariant.
\item Effort is measured by $\mathcal{E} = \int \|\dot{u}\|_G^2 \, dt$.
\end{enumerate}
\end{assumption}

\subsection{Tangential Action Spaces (TAS)}\label{subsec:TAS}

\begin{definition}[Tangential Action Space]\label{def:TAS}
A \emph{Tangential Action Space} (TAS) consists of a triple of smooth manifolds
\[
(P,C,I)
\]
together with \emph{surjective submersions of constant rank}
\[
\Phi:P\longrightarrow C, 
\qquad 
\Psi:C\longrightarrow I,
\]
such that:

\begin{itemize}
  \item $(P,G)$ is an $m$‑dimensional Riemannian manifold (the \textbf{physical manifold}) equipped with metric $G$.
  \item $C$ is an $n$‑dimensional manifold (the \textbf{cognitive manifold}).
  \item $I$ is a $k$‑dimensional manifold (the \textbf{intentional manifold}).
\end{itemize}
All manifolds are assumed Hausdorff, second countable and $C^\infty$.
The dimensions satisfy $m\!\ge\! n\!\ge\! k$.  

We impose the following standing hypothesis, needed for all later constructions:
\emph{$\Phi$ (and analogously $\Psi$) admits local smooth trivialisations,
i.e.\ $(P,C,\Phi)$ is a smooth fibre bundle with typical fibre $F$
of dimension $m-n$, and $(C,I,\Psi)$ is a smooth fibre bundle with
typical fibre of dimension $n-k$.}

\medskip
\noindent
\textbf{Practical factorisation.}  
In many examples $P$ splits as a direct product
\[
P \;=\; \overline{P} \times H,
\]
where $\overline{P}$ contains the \emph{visible} body–environment coordinates and
$H$ (often low‑dimensional) represents internal or actuator states that do
\emph{not} influence perception.  The projection then separates as
\[
\Phi(\bar p,h)=\phi(\bar p), 
\qquad 
\ker D\Phi = T_h H ,
\]
so that $\phi:\overline{P}\!\to\!C$ is a local diffeomorphism whenever $m=n$,
while $H$ realises the vertical bundle for $m>n$.

\end{definition}

\begin{figure}[h]
\centering
\includegraphics[width=\textwidth]{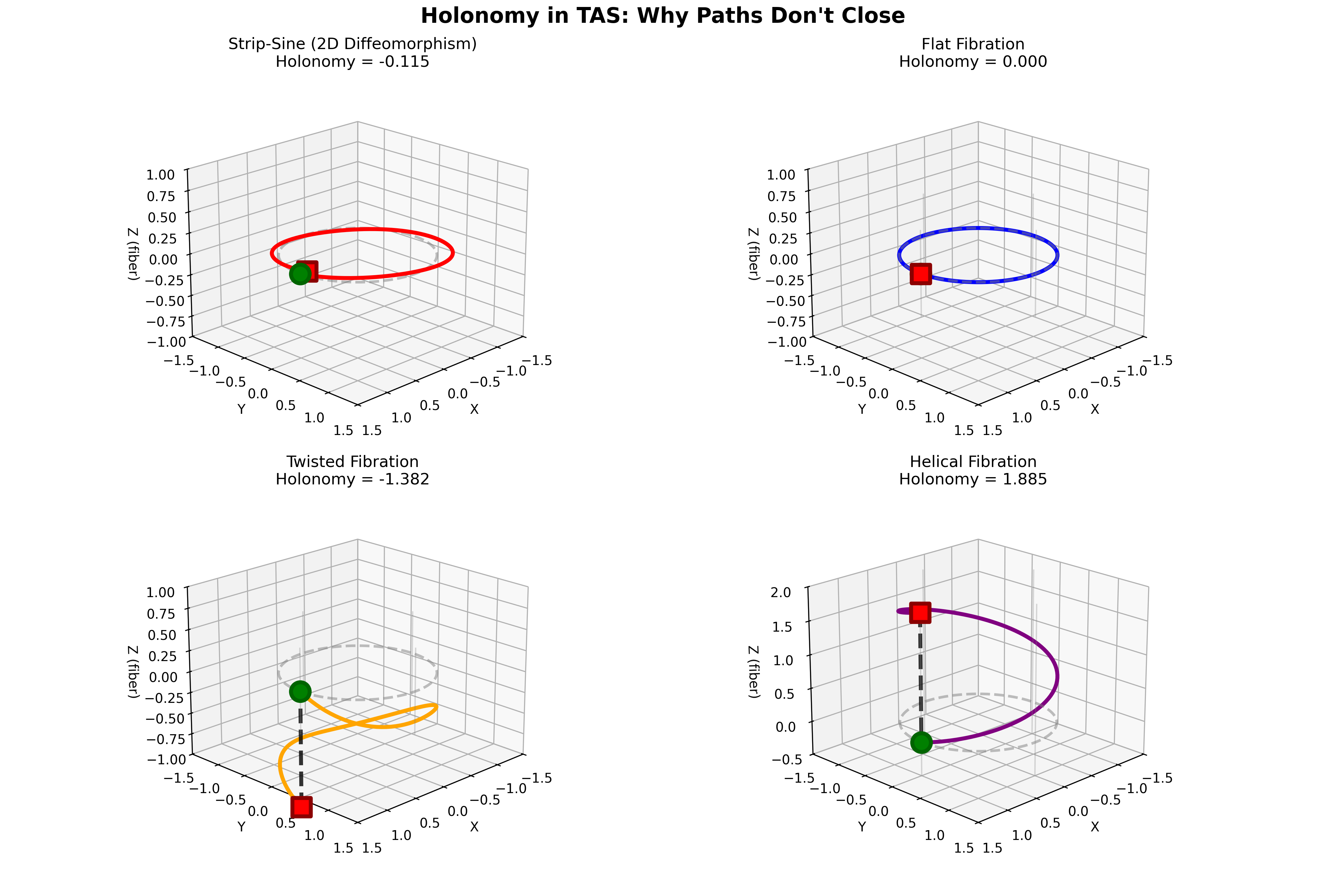}
\caption{\textbf{Holonomy in TAS: why paths do not close.} When a closed loop in cognitive space $C$ is lifted to physical space $P$, the resulting path may not close. This gap, called holonomy, encodes path-dependent memory. Green circles mark start points, red squares mark end points, and dashed lines show the holonomy. The four panels illustrate different types of systems: (Top left) A system where visible coordinates fail to close despite following the same cognitive loop. (Top right) A system with perfect closure, showing zero holonomy. (Bottom left) A system where the projection appears to close but hidden coordinates accumulate displacement. (Bottom right) A system with systematic drift, creating a helical trajectory. These examples preview how different geometric structures lead to different memory effects, which we analyse in detail in \Cref{sec:examples}.}
\label{fig:holonomy}
\end{figure}

\subsection{Biological Interpretation and Motor Control}

The projections $\Phi:P\!\to\!C$ and $\Psi:C\!\to\!I$ implement successive layers of abstraction fundamental to embodied cognition. Consider reaching for an object:

\begin{itemize}
\item \textbf{Physical space $P$:} the complete state including all muscle activations, joint angles, tendon tensions and proprioceptive signals.

\item \textbf{Cognitive space $C$:} task-relevant variables such as hand position, grip aperture and velocity, but also internal representations like predicted object properties, learned motor schemas, and world models that guide action selection.

\item \textbf{Intentional space $I$:} the goal state, such as ``grasp the cup.''
\end{itemize}

The projection $\Phi$ performs dimensional reduction, solving the degrees-of-freedom problem \cite{bernstein1967coordination}. Multiple physical configurations (different muscle activation patterns, joint trajectories) map to the same cognitive state, reflecting how biological systems extract task-relevant features while discarding irrelevant physical details. The projection $\Psi$ provides further abstraction where diverse cognitive strategies collapse into unified goals.

The lift operations reverse this abstraction hierarchy. Given a desired cognitive change (move hand to cup), the lift determines which specific muscle activations to choose. The geometric structure of the lift determines both the energetic cost and whether the movement exhibits history-dependent effects (motor memory).

\subsection{Lift Operations and Their Geometric Nature}\label{subsec:lifts}

\begin{definition}[Lift Operation]\label{def:lift}
Let $\Phi : P \to C$ be the physical to cognitive projection and denote by  
\[
\pi : \Phi^{*}TC \;\longrightarrow\; P , 
\qquad 
\Phi^{*}TC \;=\; \bigl\{\, (p,\Delta c) \mid p\in P,\; \Delta c \in T_{\Phi(p)}C \bigr\}
\]
the pull‑back tangent bundle, whose fibre over $p\in P$ is $T_{\Phi(p)}C$.
A \emph{lift operation} is a smooth bundle morphism
\[
\mathcal{L} : \Phi^{*}TC \;\longrightarrow\; TP ,
\qquad 
(p,\Delta c) \;\longmapsto\; \mathcal{L}_{p}(\Delta c)
\]
that
\begin{enumerate}
    \item \textbf{covers the identity on $P$}:  
          $\tau_{P}\!\circ\!\mathcal{L} = \pi$,  
          where $\tau_{P}:TP\to P$ is the tangent‑bundle projection; and
    \item \textbf{satisfies the projection constraint}: for every $(p,\Delta c)\in\Phi^{*}TC$,
          \[
          D\Phi_{p}\!\bigl(\mathcal{L}_{p}(\Delta c)\bigr) \;=\; \Delta c .
          \]
\end{enumerate}
Here $D\Phi_{p}:T_{p}P\to T_{\Phi(p)}C$ is the differential of $\Phi$ at $p$.
\end{definition}

\subsubsection{Diffeomorphisms: the unique geometric lift}\label{subsubsec:lift-diffeo}

When $\Phi$ is a local diffeomorphism ($m = n$), the Jacobian $D\Phi$ is invertible. This leads to a unique solution for the lift that satisfies the projection constraint:

\begin{proposition}[Unique geometric lift]\label{prop:unique_lift}
If $\Phi: P \to C$ is a local diffeomorphism, then there exists a unique lift operation $\mathcal{L}_{\text{geom}}$ given by:
\begin{equation}
\mathcal{L}_{\text{geom}}(\Delta c) = \Delta u^{\text{geom}} = D\Phi^{-1} \,\Delta c .
\end{equation}
This lift satisfies:
\begin{enumerate}
    \item It is the unique solution to the projection constraint.
    \item It preserves the norm induced by the pull‑back metric in the sense that $\|\Delta u^{\text{geom}}\|_G = \|\Delta c\|_{(D\Phi)^{*} G}$ where $(D\Phi)^{*} G$ is the pull‑back metric on $C$.
    \item It produces no holonomy for any closed loop in $C$.
\end{enumerate}
\end{proposition}
\begin{corollary}[No holonomy under (local) diffeomorphisms]
If $\Phi:P\to C$ is a (local) diffeomorphism on a simply connected neighborhood containing
the lifted trajectory, then the geometric lift of any closed loop in $C$ is a closed loop in $P$.
\end{corollary}
\begin{proof}
Same argument as above: $u(t)=\Phi^{-1}(c(t))$ closes, and in the local case use a single
branch of $\Phi^{-1}$ and uniqueness of solutions to $\dot u=D\Phi(u)^{-1}\dot c$.
\end{proof}

This is the geometric lift. It is unique and corresponds to the most energy‑efficient path among diffeomorphic lifts. Any closed loop in cognitive space results in a closed loop in physical space, so there is no intrinsic geometric memory.

\begin{definition}[Prescribed dynamics]
Let \(P=\overline{P}\times H\) as introduced in \cref{def:TAS}, where
\(\ker D\Phi = T_h H\) and the reduced map
\(\phi:\overline{P}\to C\) is a local diffeomorphism.

A \emph{prescribed dynamics} is a smooth vector field
\[
   X : P \times TC \;\longrightarrow\; TP ,
   \qquad (\bar p,h,\dot c)\mapsto X(\bar p,h,\dot c),
\]
that satisfies the projection condition
\[
   D\Phi\bigl(X(\bar p,h,\dot c)\bigr)=\dot c ,
\]
but whose component along the hidden fibre is not identically
zero, i.e.
\[
   \pr^V_{(\bar p,h)} X(\bar p,h,\dot c)\neq 0
   \quad\text{for some }(\bar p,h,\dot c).
\]

Consequently, the visible part
\(\operatorname{pr}_{T_{\bar p}\overline{P}} X\) coincides with the unique
geometric lift, while the vertical \(H\)‑component is free to accumulate
holonomy.
\end{definition}

The physical evolution is governed by the differential equation
\[
\dot p(t)=X\bigl(p(t),\dot c(t)\bigr), \qquad
\text{with } p=(\bar p,h)\in P,
\]
equivalently,
\[
\dot{\bar p}(t)=\mathcal L_{\mathrm{geom}}\bigl(\dot c(t)\bigr),\qquad
\dot h(t)=\pr^V_{(\bar p,h)} X\bigl(\bar p(t),h(t),\dot c(t)\bigr).
\]
To achieve path dependence in such systems, one must impose prescribed dynamics: rules of the form $\dot{u} = f(u, c, \dot{c})$ that determine the physical trajectory. To create memory, these dynamics must deliberately deviate from the unique geometric lift. This deviation is the source of both holonomy and an associated energetic cost.

\subsubsection{Fibrations: a space of lifts}\label{subsubsec:lift-fibration}

For fibrations ($m > n$), the equation $D\Phi(\Delta u) = \Delta c$ is underdetermined. The solution space for $\Delta u$ is an affine subspace of dimension $m - n$. Selecting a specific lift from this space requires additional structure, which is naturally provided by an Ehresmann connection.

\begin{definition}[Ehresmann connection]\label{def:ehresmann}
An Ehresmann connection \cite{ehresmann1951connexions,kobayashi1996foundations} on the fibration $\Phi: P \to C$ is a smooth distribution of horizontal subspaces $H_p \subset T_pP$ such that:
\begin{enumerate}
    \item $T_pP = H_p \oplus V_pP$ where $V_pP = \ker D\Phi_p$,
    \item $H_p$ varies smoothly with $p$,
    \item $D\Phi|_{H_p}: H_p \to T_{\Phi(p)}C$ is an isomorphism.
\end{enumerate}
\end{definition}

\paragraph{Vertical/Horizontal projections.}\label{par:vhproj}
At each $p\in P$, let $V_pP:=\ker D\Phi_p$ and let $H_p$ be the chosen horizontal space.
With respect to the metric $G$, write the $G$-orthogonal decomposition
$T_pP = H_p \oplus V_pP$ and denote the associated projections by
$\pr^H_p:T_pP\to H_p$ and $\pr^V_p:T_pP\to V_pP$.
In the special product case $P=\overline{P}\times H$ with $\Phi(\bar p,h)=\phi(\bar p)$,
we have $V_{(\bar p,h)}P \cong T_hH$ and $\pr^V_{(\bar p,h)}$ restricts to
$\pr^V_{(\bar p,h)}$ used below.

A canonical choice of connection is the metric connection, where the horizontal subspace is the $G$‑orthogonal complement of the vertical subspace. The corresponding lift is the metric lift.

\begin{proposition}[Metric lift]\label{prop:metriclift}
For a fibration $\Phi: P \to C$ with Riemannian metric $G$ on $P$, the metric lift is
\begin{equation}
\mathcal{L}_{\text{metric}}(\Delta c) 
= \Delta u^{\text{metric}} 
= G^{-1}D\Phi^\transpose \bigl[D\Phi\, G^{-1}D\Phi^\transpose\bigr]^{-1} \Delta c .
\end{equation}
Among all $\Delta u$ satisfying $D\Phi(\Delta u)=\Delta c$, $\Delta u^{\text{metric}}$ uniquely minimises $\|\Delta u\|_{G}$.
\end{proposition}

\begin{lemma}[Positive definiteness]
If $D\Phi_p:T_pP\to T_{\Phi(p)}C$ has full row rank $n=\dim C$ and $G$ is positive definite,
then $M:=D\Phi_p\,G^{-1}D\Phi_p^\top\in\mathbb R^{n\times n}$ is symmetric positive definite.
\end{lemma}
\begin{proof}
For $y\ne 0$, $y^\top M y=(D\Phi_p^\top y)^\top G^{-1}(D\Phi_p^\top y)>0$
because $D\Phi_p^\top y\ne 0$ by surjectivity of $D\Phi_p$ and $G^{-1}>0$.
\end{proof}

\begin{remark}[Pythagorean decomposition]\label{rem:pythagorean}
Let $\Delta u$ be any lift and write $\Delta u=\Delta u^{\text{metric}}+v$ with $v\in \ker D\Phi$. Then $\ip{\Delta u^{\text{metric}}}{v}_{G}=0$ and
\[
\|\Delta u\|_{G}^{2}=\|\Delta u^{\text{metric}}\|_{G}^{2}+\|v\|_{G}^{2}.
\]
Thus, any deviation from the metric lift appears as a vertical component whose squared $G$‑norm is the instantaneous excess cost. 
Since $\Delta u^{\text{metric}} = G^{-1}D\Phi^\top[D\Phi G^{-1}D\Phi^\top]^{-1}\Delta c$, for any $v \in \ker D\Phi$:
$$\langle \Delta u^{\text{metric}}, v \rangle_G = v^\top D\Phi^\top[D\Phi G^{-1}D\Phi^\top]^{-1}\Delta c = 0$$
since $D\Phi v = 0$.
\end{remark}

\section{Related Work}
\subsection{Geometric Approaches to Embodied Cognition}
Differential geometric methods have long been applied to model embodied agents and robotic systems \cite{murray1994mathematical,bloch2003nonholonomic}. Robotic configurations are naturally described on smooth manifolds, and tools such as fibre bundles and gauge connections have been used to analyse their motion \cite{marsden1999introduction,kobayashi1996foundations}. For example, cyclic changes in a robot's "shape" coordinates can lead to net movements – a geometric phase \cite{berry1984quantal,wilczek1984appearance,simon1983holonomy} – even when the system returns to its initial shape. Classic instances include the parallel parking problem, where oscillatory steering yields a sideways displacement, and the locomotion of snake-like robots \cite{bullo1999kinematic,ostrowski1998geometric} or swimming microrobots modeled using fibre bundle formalisms. Montgomery's gauge-theoretic analysis \cite{montgomery1993gauge} showed how systems with internal shape variables can achieve net displacements through cyclic deformations, treating configuration space connections analogously to Yang–Mills fields. Shapere and Wilczek \cite{shapere1989geometric} demonstrated how deformable bodies exploit gauge connections for locomotion. These geometric approaches confirm that path-dependent effects (holonomy) play a role in purely physical systems. However, prior work has largely focused on the mechanics and control of movement itself, rather than incorporating cognitive states or memory.

Tangential Action Spaces (TAS) extend this line of work by introducing explicit cognitive and intentional manifolds on top of the physical manifold. This allows us to ask new questions – for instance, how the geometry of the perception-action map $\Phi$ influences an agent's internal memory of past actions – which were not addressed in earlier geometric frameworks. While gauge theories and fibre bundle models provide the mathematical foundation (connections, horizontal lifts, etc.), they have not previously been used to unite energy costs with path-dependent cognitive effects. TAS builds on these geometric insights and brings in the novel consideration of a cost–memory trade-off, something absent in prior geometric approaches to embodied cognition.

\subsection{Dynamical Systems and Enactivism}
Our framework is also informed by dynamical systems theory (DST) and the enactive approach in cognitive science. Enactivism posits that cognition arises through a dynamic interaction between an acting organism and its environment, emphasizing the idea of structural coupling, the continual mutual influence between agent and world. Classic enactive theory \cite{varela1991embodied} and related DST models describe agents as dynamical systems coupled to their surroundings, explaining cognition as an emergent, history-dependent process rather than a sequential computation. This perspective resonates with Gibson's ecological approach to perception \cite{gibson1979ecological}, which emphasizes direct perception through agent-environment interaction rather than internal representation. For example, Beer's agent-based models \cite{beer1995dynamical} and Thelen's work on infant motor development use systems of differential equations to capture how cognitive behaviour unfolds over time in tandem with bodily action. These approaches compellingly illustrate phenomena like sensorimotor contingencies and limit-cycle behaviors in agent–environment systems.

However, they often lack a geometric formalism: the state space dynamics are usually described abstractly, without an underlying manifold structure that differentiates between "physical" and "cognitive" coordinates. TAS contributes a formal geometric scaffolding to the DST/enactive perspective. The projection $\Phi: P \to C$ in TAS is a concrete realization of structural coupling; it mathematically encodes how physical states map to cognitive states. By doing so, TAS makes it possible to apply differential geometric tools (such as connections and holonomy) to analyse classic enactive concepts. For instance, where enactivism might qualitatively discuss how an agent's history of sensorimotor interaction can alter its cognitive state, TAS can quantify this as path-dependent parallel transport on a fibre bundle (yielding measurable holonomy). In essence, TAS enriches dynamical systems accounts with a fibre bundle geometry: dynamical trajectories become lifted paths on manifolds. This added structure lets us identify when a system's behaviour is equivalent to a gradient flow on a potential vs. when it exhibits truly path-dependent evolution (something enactive accounts acknowledge conceptually but do not formalize). Thus, TAS complements DST and enactivist models by offering a unifying geometric language – one that preserves their insights about coupling and emergence, but adds the ability to rigorously distinguish conservative (path-independent) dynamics from nonconservative (history-dependent) dynamics.

\subsection{Energy and Efficiency in Motor Control}
A separate line of relevant work comes from optimal motor control and principles of efficient movement. In both biomechanics and robotics, it has been widely observed that biological motions often optimize some cost functional leading to models like minimum-jerk trajectories for human reaching, minimum torque-change and minimum energy expenditure for multi-joint movements. For example, Flash and Hogan's minimum-jerk model \cite{flash1985coordination} accounted for the straight-line hand paths and smooth velocity profiles seen in reaching movements by assuming the nervous system minimizes the jerk (third derivative of position) integrated over the movement duration. Similarly, Uno et al. \cite{uno1989formation} proposed a minimum torque-change criterion for multi-joint arm motions, and Alexander \cite{alexander1997minimum} hypothesized that human arm trajectories minimise metabolic energy cost. These optimality principles have been very successful in explaining and predicting kinematic patterns in tasks like pointing, locomotion, and gaze control. They also align with optimal control formulations in robotics (e.g., generating trajectories that minimise integrated squared torque or energy).

What these approaches typically do not address, however, is path-dependent memory or hysteresis. The optimization is usually performed per movement, from an initial state to a target, without considering the internal state memory of how that movement was executed. In other words, a minimum-jerk trajectory is optimal for that reach, but if the same reach is repeated, the model doesn't predict any difference based on the previous attempt's path. There is no notion that taking a different path to the same end point could leave an agent in a different internal state. By contrast, our TAS framework explicitly studies scenarios where the *same end-point* in $C$ or $P$ can be reached via different paths with different energetic costs and different retained memories (holonomies). Prior motor control models also typically assume a fixed mapping from desired task outcome to motor commands, whereas TAS reveals that when the $\Phi: P\to C$ mapping has nontrivial geometry, there can be multiple lifts (action policies) that achieve the same nominal behaviour with different energy expenditures.

In that sense, TAS bridges a gap between efficiency and memory: it generalizes the efficiency-centric view of optimal control by showing how striving for efficiency (e.g. following the metric-minimizing geodesic lift) conflicts with the introduction of path-based memory. Our results resonate with the intuition behind minimum-energy and minimum-effort models – indeed, the metric lift in TAS is analogous to the energy-optimal trajectory but we additionally pinpoint the energetic cost of deviating from that optimal path in order to encode memory. Thus, TAS can be seen as an extension to optimal motor control theory: one that incorporates the internal-state consequences of trajectory choices, not just their immediate energetic cost.

\subsection{Memory and Holonomy in Physical Systems}
The notion that physical systems can "remember" how they moved, independent of their start and end points, is well established in fields like classical mechanics and quantum physics, typically under the banner of geometric phase \cite{berry1984quantal,hannay1985angle,simon1983holonomy} or holonomy. Classical examples include robotic locomotion where cyclic gaits produce net displacements, and deformable bodies that achieve motion through shape changes. In the quantum realm, Berry's phase \cite{berry1984quantal} (and its non-Abelian generalizations \cite{wilczek1984appearance,simon1983holonomy}) shows that a quantum system slowly driven around a closed loop in parameter space acquires a phase shift dependent only on the loop's geometry, not on time or energy expended – effectively, a memory of the path taken.

These diverse examples illustrate that holonomy is a unifying concept: it appears whenever the state space has nontrivial curvature or topology. What has been lacking is a connection of these insights to cognitive and control processes. TAS provides that connection by treating an agent's cognitive state as analogous to a "position on a base manifold" and its physical state as a point in a higher-dimensional fibre space. In doing so, TAS allows us to interpret classical geometric phases as instances of embodied memory. For example, the Berry phase becomes a special case of TAS holonomy where the "cognitive manifold" is the parameter space and the physical effect is a phase shift. The gauge theories of locomotion \cite{kelly1995geometric,ostrowski1998geometric} (e.g. parallel parking, snake robot gaits) become TAS scenarios where $C$ is the shape space and $P$ includes the position/orientation. Here the holonomy in $P$ corresponds to locomotion. By unifying these under one framework, we can compare and contrast the energy costs of different types of geometric memory. Prior studies of geometric phase generally considered idealized systems and often assumed lossless, conservative dynamics (to cleanly observe the phase effect). TAS broadens this by examining non-conservative cases (curved connections with energy dissipation) and explicitly asking: what is the energetic price of acquiring a given holonomy? In summary, this subsection of related work underscores that the TAS framework synthesizes themes from geometric phase theory and holonomy; it bridges physical and informational memory.

\subsection{Predictive Processing and Active Inference}
Another influential framework in cognitive science and neuroscience is the family of theories around predictive processing, including the Free Energy Principle and Active Inference. These theories propose that intelligent agents \cite{friston2010free,friston2010action}(brains, robots, or organisms) operate by constantly predicting their sensory inputs and updating their internal beliefs to minimise prediction error or "surprise." In the Free Energy Principle formulation, an agent is said to minimise a variational free-energy bound on surprise by adjusting both its internal neural states and its actions. This leads to a picture of behaviour where perception and action are in service of reducing prediction errors: perception updates the internal model to better fit sensory data, while action changes the world (or the agent's sensory input) to better fit the predictions.

Active Inference, in particular, extends this idea to action selection, asserting that agents select motor commands that are expected to minimise future prediction errors (often framed as fulfilling prior expectations about desired states). These ideas have been extremely powerful in explaining everything from reflexes to high-level cognitive biases, effectively unifying homeostatic regulation, perception, and goal-directed behaviour under one normative principle.

However, predictive processing models usually abstract away the detailed geometry of the physical world. They are typically formulated in terms of probabilistic state estimates and do not say much about manifolds, curvature, or path integrals. The "embodiment" of predictive coding is acknowledged (e.g. Friston's principle is touted as an explanation of embodied perception–action loops), but the formalism tends to lump all physical interactions into a generic probabilistic mapping (the generative model and likelihood function). As a result, concepts like energy cost or path-specific memory are not explicitly represented. For instance, an active inference agent might infer a policy that keeps it in a safe state, but standard formulations won't account for the fact that two policies reaching the same state might expend different energy or leave different internal residues.

TAS can provide a valuable geometric grounding for these ideas. By mapping the abstract variables of a generative model onto $P$, $C$, and $I$ manifolds, we can interpret prediction error minimization in terms of movements along those manifolds. Notably, TAS introduces the idea that there is a geodesic (energy-optimal) way to realize a given prediction or goal-directed change, namely the metric lift, and that deviating from this geodesic corresponds to the agent encoding some additional information (holonomic memory). In a predictive processing context, this suggests a refinement: agents not only minimise surprise, but they may do so in ways that either conserve or expend extra energy depending on whether they also need to learn or memorize something from the experience. For example, if an agent repeatedly predicts and perceives a certain outcome, predictive coding alone might adapt its expectations; TAS would add that if the agent's internal manifold allows multiple lifts, it could either follow a habit (energy-efficient repetition) or explore a new lift (higher cost, but yielding learning of a novel sensorimotor mapping). In Active Inference terms, one typically defines a free-energy minimizing policy without detailing the path geometry; TAS could help characterize which policy among those that achieve a given outcome is geometrically natural (minimal energy) versus which involve detours that create lasting state changes. In summary, predictive processing and Active Inference provide a high-level normative target (minimise prediction error/free energy) for adaptive behaviour, and TAS complements this by revealing the underlying geometric mechanics needed to implement such behaviour in an embodied agent.

By doing so, TAS links the thermodynamic and information-theoretic efficiency emphasized by the Free Energy Principle with the physical energy efficiency (and memory trade-offs) emphasized in our framework. This not only grounds predictive processing in a concrete embodiment but also highlights scenarios where informational efficiency and energetic efficiency may conflict, for instance, when gaining information (reducing uncertainty) requires taking an energetically costly path. Such insights are largely outside the scope of traditional predictive coding models, but arise naturally in the TAS perspective, suggesting fertile ground for integrating the two frameworks in future work.

Having established the geometric foundations, we now turn to the energetic consequences of different lift choices.

\section{Travel Cost and the Price of Memory}\label{sec:cost}

\subsection{Energetic Foundations}

The energetic cost of executing a physical trajectory $u(t)$ along a path $\gamma$ is given by the integrated squared velocity:
\begin{equation}
\mathcal{E}[\gamma] = \int_\gamma \|\dot{u}(t)\|_{G}^{2} \, dt .
\end{equation}

\begin{lemma}[Optimality of metric lift]\label{lem:optimality}
For any fibration $\Phi: P \to C$ and any lift operation $\mathcal{L}$ yielding $\Delta u$ with $D\Phi(\Delta u) = \Delta c$,
\begin{equation}
\|\Delta u^{\text{metric}}\|_G \leq \|\Delta u\|_G ,
\end{equation}
with equality if and only if $\Delta u = \Delta u^{\text{metric}}$.
\end{lemma}

\begin{proof}
Any $\Delta u$ satisfying $D\Phi(\Delta u) = \Delta c$ can be written as $\Delta u = \Delta u^{\text{metric}} + v$ where $v \in \ker D\Phi$. By the Pythagorean decomposition (Remark~\ref{rem:pythagorean}), $\|\Delta u\|_G^2 = \|\Delta u^{\text{metric}}\|_G^2 + \|v\|_G^2 \geq \|\Delta u^{\text{metric}}\|_G^2$, with equality iff $v = 0$.
\end{proof}

\subsection{The Cost–Memory Duality}

A fundamental relationship links path‑dependent memory (holonomy) and energetic cost. Any lift that generates holonomy must deviate from the minimal‑energy metric lift, and the deviation has an energetic price.

\begin{theorem}[Cost–memory trade-off]\label{thm:cost-memory}
Let $\gamma:[0,T]\to C$ be $C^1$.
\begin{enumerate}
    \item (Diffeomorphism) If $P=C$ and a prescribed lift $L_{\mathrm{prescr}}$ satisfies
$L_{\mathrm{prescr}}(\dot c)=L_{\mathrm{geom}}(\dot c)+v_{\mathrm{vert}}$ with
$v_{\mathrm{vert}}\neq 0$ on a set of positive measure, then
$\mathcal E_{\mathrm{prescr}}[\gamma]>\mathcal E_{\mathrm{geom}}[\gamma]$.
Since $\|v_{\text{vert}}\|^2 > 0$ on a set $S$ of positive measure and the integrand is continuous, $\int_S \|v_{\text{vert}}\|^2 dt > 0$.
\item (Fibration) For any connection with horizontal map $L_H(\dot c)$,
$\|L_H(\dot c)\|_G^2=\|L_{\mathrm{metric}}(\dot c)\|_G^2+\|v_{\mathrm{vert}}(\dot c)\|_G^2$
pointwise, hence $\mathcal E_H[\gamma]\ge\mathcal E_{\mathrm{metric}}[\gamma]$,
with equality iff $v_{\mathrm{vert}}\equiv 0$.
\end{enumerate}

\end{theorem}

\begin{proof}
Decompose any feasible velocity as $L_{\mathrm{metric}}(\dot c)+v_{\mathrm{vert}}$ with
$v_{\mathrm{vert}}\in V_pP$; orthogonality in $G$ (Remark~\ref{rem:pythagorean}) gives
the stated energy identities and strictness when $v_{\mathrm{vert}}\not\equiv0$.
\end{proof}

\begin{remark}[Physical interpretation]
The cost-memory trade-off theorem establishes that memory is never free: any path-dependent behaviour requires excess energy above the geometric minimum. This excess is precisely the energy stored in vertical motions that create holonomy.
\end{remark} 
\begin{proposition}[Small‑loop law]\label{prop:smallloop}
Let $\nabla$ be any smooth connection on $\Phi:P\to C$ with curvature two‑form $F$. For a sufficiently small closed loop $\gamma$ bounding area $A$ in $C$, the induced holonomy $\Delta u_{\mathrm{vert}}$ satisfies $\|\Delta u_{\mathrm{vert}}\| \leq K\,|A|$ for some local constant $K$. Moreover, the excess effort over the metric lift obeys
\[
\mathcal E_{\nabla}[\gamma]-\mathcal E_{\text{metric}}[\gamma] \;\ge\; \kappa\, \|\Delta u_{\mathrm{vert}}\|^{2} + o(A^{2})
\]
for a positive constant $\kappa$ depending on $(G,\nabla)$ at the base point.
\end{proposition}

\begin{proof}
Assume Abelian fibre. Let $\omega$ be the connection $1$-form and $F=d\omega$.
For a small loop $\gamma$ bounding $S_\gamma$ inside a coordinate ball $U$,
$\Delta u_{\mathrm{vert}}=\iint_{S_\gamma}F$. Hence
$\|\Delta u_{\mathrm{vert}}\|\le \|F\|_{L^\infty(U)}\,\mathrm{Area}(S_\gamma)=:K\,|A|$.

Write the vertical velocity as $v(t)=\omega(\dot u_H(t))$ along the horizontal lift;
then $\Delta u_{\mathrm{vert}}=\int_0^T v(t)\,dt$. By Cauchy–Schwarz,
$\int_0^T\|v\|_G^2\,dt \ge \|\Delta u_{\mathrm{vert}}\|_G^2/T$.
Fix a time reparametrization with $T=1$ for small loops; then
$\mathcal E_H[\gamma]-\mathcal E_{\mathrm{metric}}[\gamma]
=\int_0^1\|v\|_G^2\,dt \ge \|\Delta u_{\mathrm{vert}}\|_G^2$.
Thus $\kappa=1$ in this parametrization and the remainder is $o(A^2)$ since
$\Delta u_{\mathrm{vert}}=O(A)$ by the first bound.
\end{proof}

\section{Holonomy and Path‑Dependent Memory}\label{sec:holonomy}

\subsection{Geometric Holonomy in Fibrations}\label{subsec:geom-hol}

For fibrations equipped with an Ehresmann connection, path dependence arises from curvature. For a connection defined by a connection one‑form $\omega$, the curvature is the two‑form $F = d\omega + \omega \wedge \omega$. For Abelian structure groups, $F=d\omega$.

\begin{definition}[Curvature and holonomy]\label{def:curvhol}
For a connection with connection one‑form $\omega$ on a fibration $\Phi: P \to C$:
\begin{enumerate}
    \item The curvature two‑form is $F = d\omega + \omega \wedge \omega$.

    \item For a closed loop $\gamma \subset C$ bounding surface $S_\gamma$, the holonomy is
\begin{equation}
\text{Hol}(\gamma)=\exp\!\big(\iint_{S_\gamma} F\big)\in \Group,
\end{equation}
where $\Group$ is the structure group of the bundle.
    \item For Abelian structure groups, this simplifies to
    \begin{equation}
    \Delta u^i_{\text{fibre}} = \iint_{S_\gamma} F^i .
    \end{equation}
\end{enumerate}
\end{definition}
For non-Abelian groups, the exponential must be path-ordered. Our assumption (A1) of Abelian structure groups ensures the integral is well-defined independent of how $S_\gamma$ is parametrized.

Nonzero curvature implies that lifting a closed cognitive loop generally yields an open physical path. As demonstrated in our examples, constant and variable curvature lead to distinct forms of geometric holonomy. Conversely, even with nontrivial topology, a flat connection ($F=0$) produces no holonomy for contractible loops.

\begin{remark}[Holonomy versus path‑dependent displacement]
We distinguish between geometric holonomy, a structure‑group element resulting from parallel transport around a closed loop, and path‑dependent displacement, the net physical displacement after traversing a closed cognitive loop. For Abelian fibrations, these coincide. For diffeomorphisms, only path‑dependent displacement via prescribed dynamics is possible.
\end{remark}

\subsection{Prescribed Holonomy in Diffeomorphisms}

For diffeomorphisms, which lack intrinsic geometric holonomy, path dependence must be explicitly engineered through prescribed dynamics. Given a control law $\dot{u} = \mathcal{L}_{\text{prescr}}(u, c(t), \dot{c}(t))$ that generates a physical trajectory $u(t)$ in response to a desired cognitive loop $c(t)$ for $t \in [0, T]$, the holonomy is
\begin{equation}
\Delta u_{\text{prescr}} = u(T) - u(0) = \int_0^T \mathcal{L}_{\text{prescr}}(u(t), c(t), \dot{c}(t)) \, dt ,
\end{equation}
and it necessarily incurs excess energy by \cref{thm:cost-memory}.
With the cost-memory relationship established, we can now systematically classify all possible path-dependent behaviors.

\section{Classification of Path‑Dependent Behaviors}

\textbf{Classification Principle}
Given $(\Phi,G)$ together with either a connection $\nabla$ (fibration case) or a prescribed
vertical field $X_{\mathrm{vert}}$ (diffeomorphic case), the induced behaviour falls primarily into one
of four archetypes: (1) intrinsically conservative, (2) conditionally conservative, (3) geometrically
nonconservative, (4) dynamically nonconservative. Mixed cases (nonzero curvature \emph{and}
nonzero $X_{\mathrm{vert}}$) can be described by the pair $(F, X_{\mathrm{vert}})$; figures use the dominant mechanism.

\noindent
\textbf{Intrinsically conservative} systems arise from diffeomorphisms with the geometric lift. They exhibit zero holonomy and minimal cost.

\noindent
\textbf{Conditionally conservative} systems occur in fibrations with flat connections ($F=0$), as shown in the flat and cylindrical examples. Despite having a fibre structure, these systems exhibit zero holonomy for contractible loops. They achieve minimal cost when using the metric lift.

\begin{figure}[H]
\centering
\includegraphics[width=\textwidth]{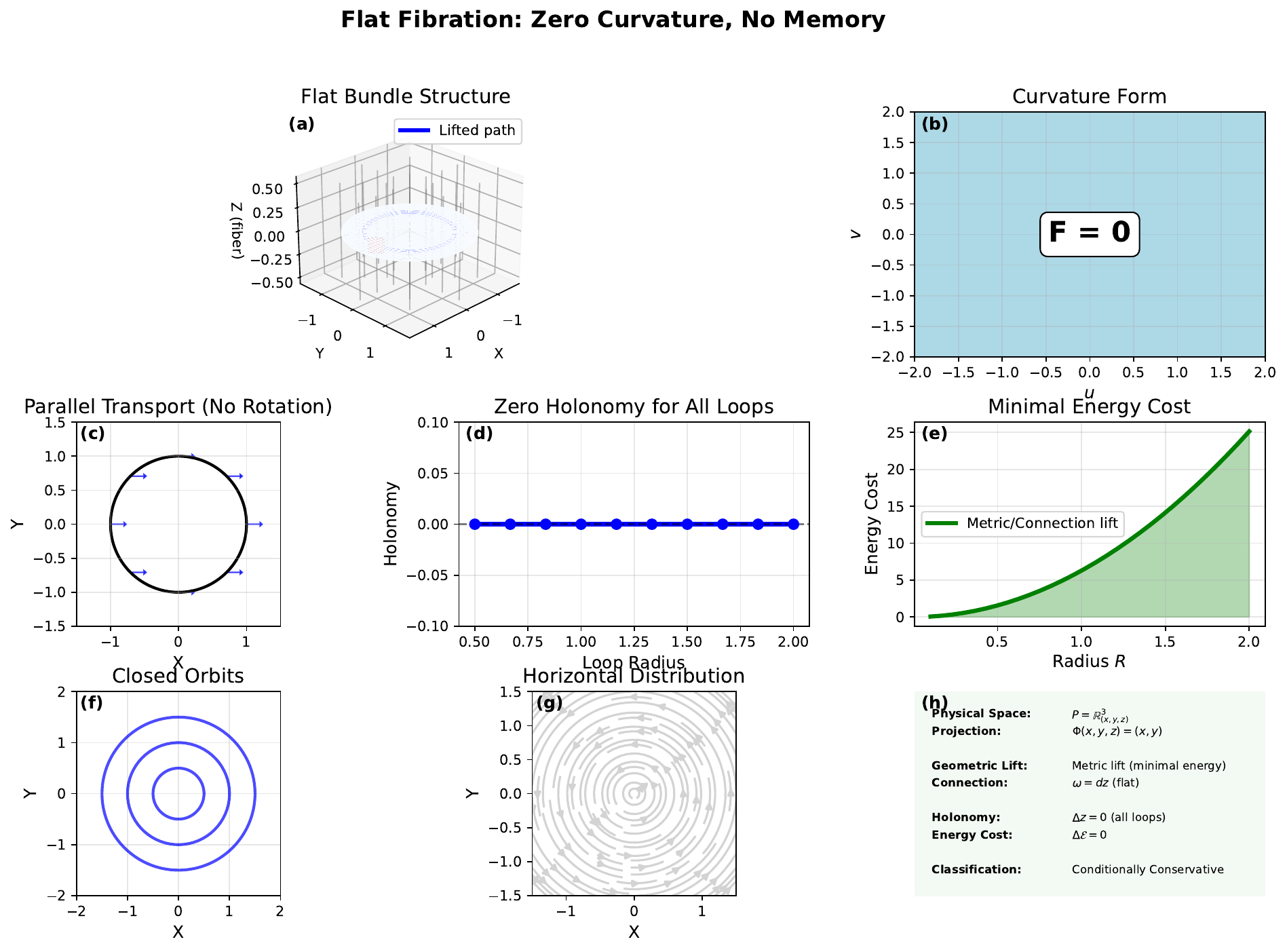}
\caption{\textbf{Flat fibration: zero curvature, no memory.} 
(a) Three‑dimensional bundle structure shows a lifted path that closes exactly, demonstrating zero holonomy despite the fibre structure. 
(b) The curvature form is $F = 0$ everywhere, confirming the flat Ehresmann connection. 
(c) Parallel transport preserves horizontal subspaces completely around closed loops. 
(d) Holonomy remains zero for all loop radii. 
(e) Energy analysis shows the metric and connection lifts coincide, resulting in no excess cost $\Delta\mathcal{E} = 0$. 
(f) Orbits for multiple cognitive trajectories remain closed. 
(g) The horizontal distribution of the flat connection. 
System classification: conditionally conservative (fibration with a flat connection).}
\label{fig:flat}
\end{figure}

\noindent
\textbf{Geometrically nonconservative} systems emerge in fibrations with curved connections ($F \neq 0$), where nonzero geometric holonomy arises from curvature. The cost depends on the specific connection chosen. Helical and twisted fibrations illustrate constant and variable curvature effects.

\noindent
\textbf{Dynamically nonconservative} systems appear in diffeomorphisms with prescribed lifts. Here, path dependence is achieved by engineering dynamics that deviate from the geometric lift, which necessarily incurs an excess cost $\Delta\mathcal{E} > 0$. The strip–sine system is a prime example.

\section{Illustrative Examples}
\label{sec:examples}
We now illustrate the theoretical framework through four canonical examples that span all classes of path-dependent behaviour. These examples progress from engineered memory in hidden dimensions to intrinsic geometric memory, demonstrating how different mathematical structures lead to qualitatively different memory mechanisms and energetic costs.

\subsection{The Strip–Sine System: engineering memory in a hidden fibre}
\label{subsec:strip-sine}

Let the physical space be the direct product
\[
  P \;=\; \mathbb R^{2}_{(u,v)} \times \mathbb R_{h},
\]
where the extra coordinate $h\in\mathbb R$ represents an internal actuator state that is not observable in perception.
The cognitive space is $C=\mathbb R^{2}_{(c_{1},c_{2})}$ and the
projection (extended trivially over~$h$) is
\begin{equation}
  \Phi(u,v,h) \;=\; (c_{1},c_{2})
               \;=\; \bigl(u,\;v+\kappa\sin u\bigr),
  \label{eq:strip-sine-projection}
\end{equation}
with differential
\[
  D\Phi_{(u,v,h)} \;=\;
  \begin{pmatrix}
     1 & 0 & 0\\[4pt]
     \kappa\cos u & 1 & 0
  \end{pmatrix},
  \qquad
  \operatorname{rank}D\Phi = 2 .
\]
Hence $\ker D\Phi = T_{h}\mathbb R_{h}$, so $H=\mathbb R_{h}$ is the
hidden fibre and $\overline{P}=\mathbb R^{2}_{(u,v)}$ the
visible sub‑manifold, in line with
\cref{def:TAS}.

\paragraph{Geometric lift.}
For any cognitive velocity $(\dot c_{1},\dot c_{2})$ the unique
(horizontal, energy‑minimal) lift in the visible coordinates is
\[
  \dot u = \dot c_{1},
  \qquad
  \dot v = \dot c_{2} - \kappa\cos u \,\dot c_{1},
  \qquad
  \dot h = 0 .
\]
Because $u$ and $v$ are related to $(c_{1},c_{2})$ by the global
diffeomorphism
$\varphi(u,v)=(u,v+\kappa\sin u)$, any closed cognitive loop
$c(t)$ maps to a closed visible loop $(u(t),v(t))$; thus the geometric
lift exhibits no holonomy in $(u,v)$.

\paragraph{Prescribed dynamics (memory in the hidden fibre).}
To store path history we keep the horizontal part unchanged and add a
vertical component:
\begin{align}
  \dot u &= \dot c_{1}, \\[4pt]
  \dot v &= \dot c_{2} - \kappa\cos u \,\dot c_{1}, \\[4pt]
  \dot h &= f(c,\dot c), \qquad
           \text{with } 
           f(c,\dot c) := \alpha\bigl(c_{1}\,\dot c_{2} - c_{2}\,\dot c_{1}\bigr),
  \label{eq:strip-sine-f}
\end{align}
where $\alpha\in\mathbb R$ is a tunable gain.
Because $D\Phi(\dot u,\dot v,\dot h)=(\dot c_{1},\dot c_{2})$ still
holds, the projection constraint is respected.

\paragraph{Holonomy in the hidden fibre.}
For a closed cognitive loop $\gamma\subset C$ we obtain
\[
  \Delta h 
  \;=\; h(T)-h(0)
  \;=\; \alpha \!\!\oint_{\gamma}\!
        \bigl(c_{1}\,dc_{2} - c_{2}\,dc_{1}\bigr)
  \;=\; 2\alpha\,\operatorname{Area}(\gamma) .
\]
Thus the internal state $h$ records the loop area, providing
a clear example of path‑dependent memory.

\paragraph{Energetic cost.}
The additional instantaneous power required is
$\|\dot h\|^{2} = \alpha^{2}
 (c_{1}\,\dot c_{2} - c_{2}\,\dot c_{1})^{2}$,
so the excess cost above the geometric lift satisfies
\[
  \Delta\mathcal E 
  \;=\; \int_{0}^{T} \!\|\dot h(t)\|^{2}\,dt
  \;\propto\; \alpha^{2}\,\bigl[\operatorname{Area}(\gamma)\bigr]^{2},
\]
exemplifying the cost–memory trade‑off predicted by
\cref{thm:cost-memory}.  Visible coordinates close perfectly,
but memory is accumulated in the hidden fibre at an energetic price
controllable via $\alpha$.

\begin{figure}[H]
\centering
\includegraphics[width=\textwidth]{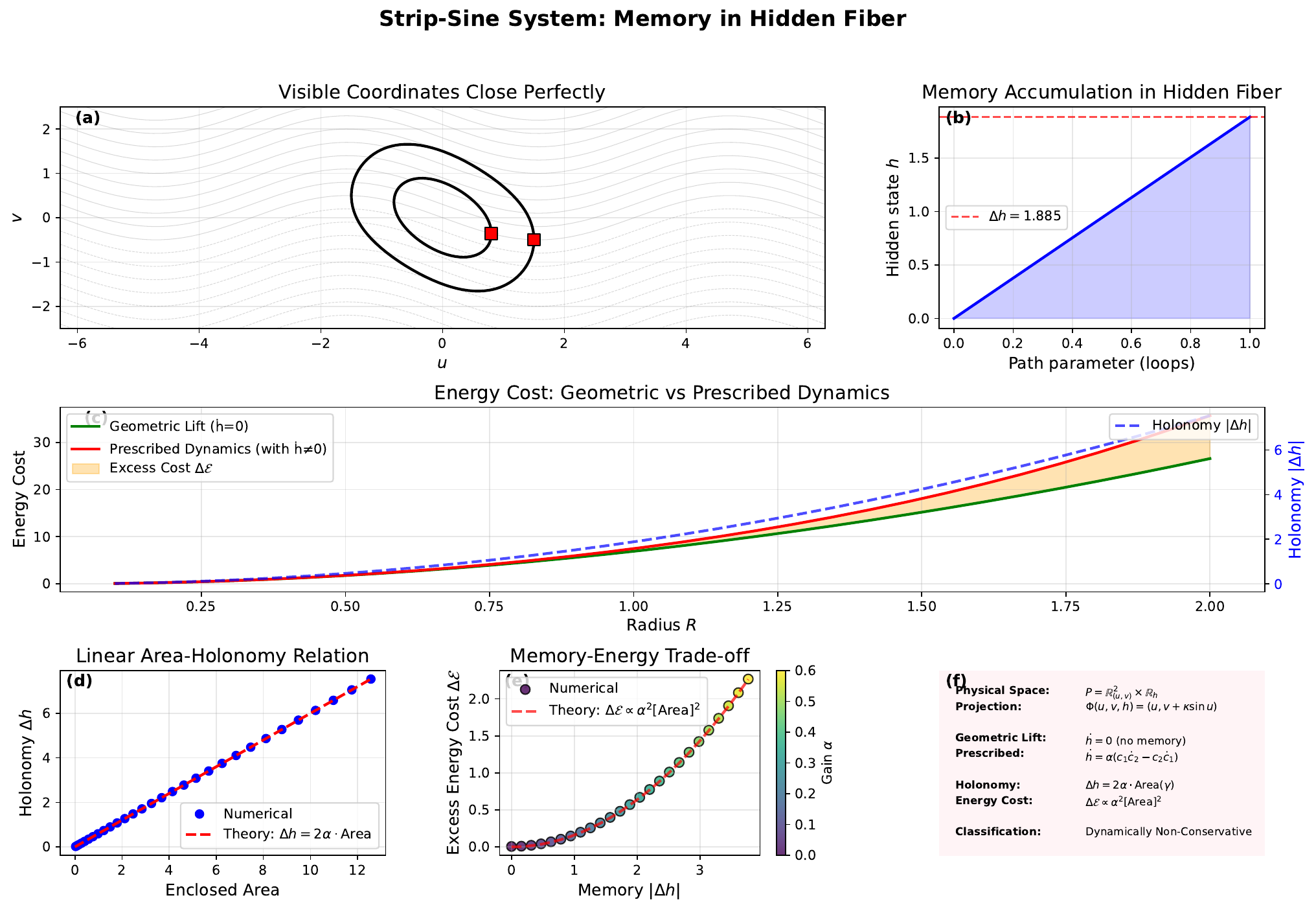}
\caption{\textbf{Strip–Sine system with area‑based memory in a hidden fibre.}
(a) Visible coordinates $(u,v)$ close perfectly for circular cognitive loops of different radii; the projection $\Phi(u,v,h)=(u,v+\kappa\sin u)$ is a diffeomorphism on $(u,v)$, so no visible holonomy occurs (grey streamlines show the geometric lift).
(b) Hidden‑state evolution for a unit circle: the prescribed dynamics $\dot h=\alpha(c_{1}\dot c_{2}-c_{2}\dot c_{1})$ integrates to $\Delta h = 2\pi\alpha$ (here $\alpha=0.3$, hence $\Delta h\!=\!1.885$).
(c) Energetic cost comparison using the squared‑speed functional $\mathcal E=\int\|\dot u\|_{G}^{2}dt$: the geometric lift (green) scales quadratically with radius $R$, whereas the prescribed dynamics (red) incurs an additional quartic term from the hidden fibre; the shaded area is the excess cost $\Delta\mathcal E\propto R^{4}$.
(d) Linear area–holonomy law: numerical data (blue) follow $\Delta h = 2\alpha\,\mathrm{Area}$ (red dashed line) exactly.
(e) Memory–energy trade‑off: simulated points collapse on the theoretical curve $\displaystyle\Delta\mathcal E=\tfrac{(\Delta h)^{2}}{2\pi}$ (dashed), confirming the quadratic cost of storing path history.
System classification: dynamically nonconservative TAS with engineered memory in a hidden fibre.}
\label{fig:strip_sine}
\end{figure}

For a circular cognitive loop of radius $R$ centred at the origin, one
finds, for the choice $f=\alpha(c_1\dot{c}_2 - c_2\dot{c}_1)$,
\[
\Delta h = 2\alpha \cdot \mathrm{Area}(\gamma) = 2\alpha\pi R^2,
\qquad
\Delta\mathcal{E} = \frac{(\Delta h)^2}{2\pi} = 2\pi\alpha^2 R^4,
\]
recovering the cost–memory scaling in
\cref{fig:strip_sine}(d,e).

\subsection{Helical Fibration: natural geometric memory}

Consider a helical fibration with physical space $P = \mathbb{R}^3(x,y,z)$ and cognitive space $C = \mathbb{R}^2(x,y)$, with projection $\Phi(x,y,z) = (x,y)$. Equip this with a connection defined by the one‑form $\omega = dz - \alpha(y \,dx - x\, dy)$. This connection has constant curvature $F = d\omega = 2\alpha \,dx \wedge dy$.

For a closed loop $\gamma$ in the cognitive plane $C$, the geometric holonomy is the integral of curvature over the enclosed area:
\begin{equation}
\Delta z = \iint_{\text{Area}(\gamma)} F = 2\alpha \cdot \text{Area}(\gamma) .
\end{equation}
This demonstrates geometric path dependence arising from curvature. The travel cost depends on $\alpha$, with larger $\alpha$ yielding more holonomy but at higher energetic cost, a concrete manifestation of the cost–memory trade‑off. The metric connection corresponds to $\alpha=0$.

\begin{figure}[H]
\centering
\includegraphics[width=\textwidth]{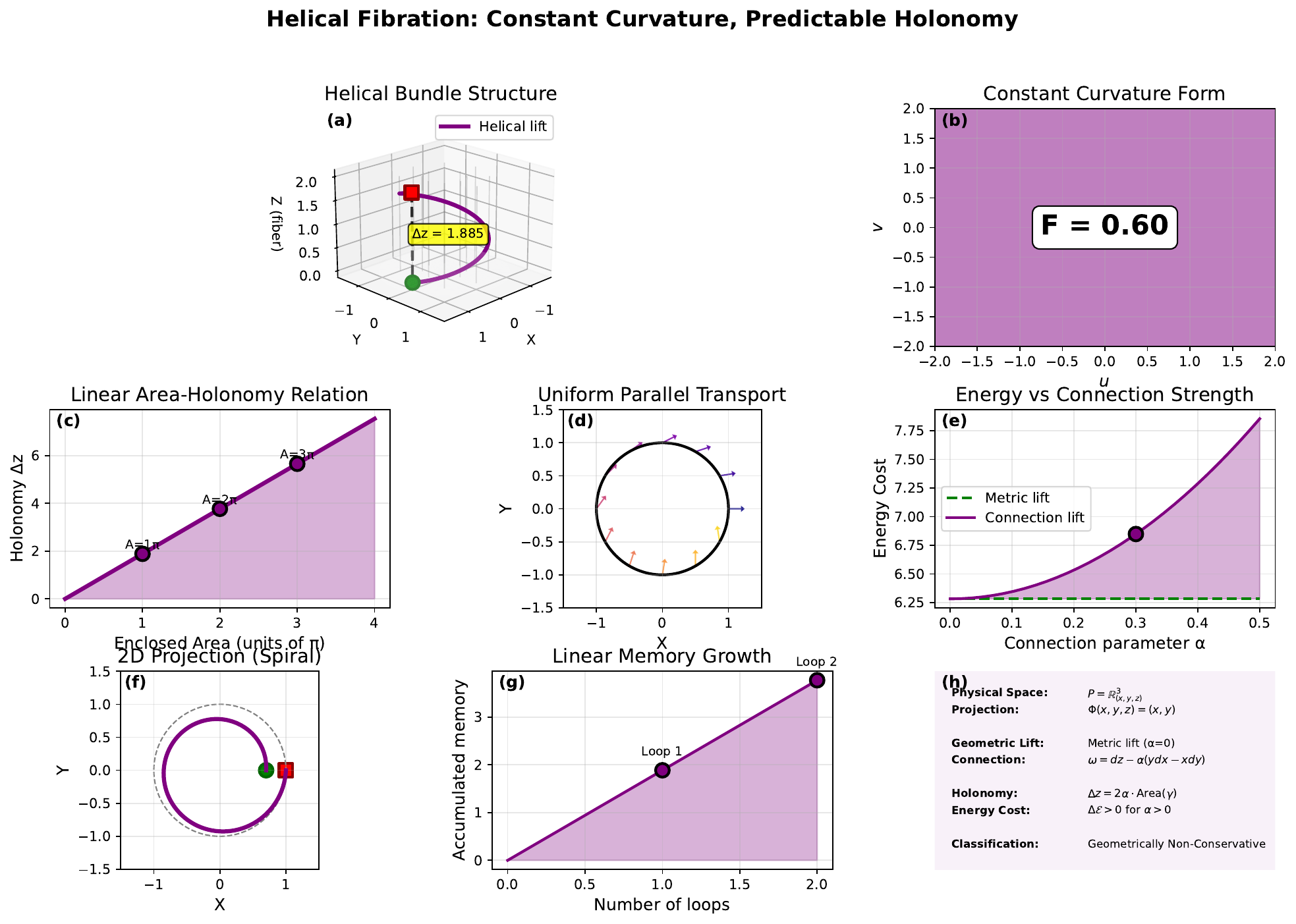}
\caption{\textbf{Helical fibration: constant curvature, predictable holonomy.} 
(a) The helical bundle structure shows a uniform rise with constant pitch, yielding a holonomy $\Delta z = 1.885$ for the unit circle path shown. 
(b) The constant curvature form is $F = 2\alpha\, dx \wedge dy$ (with $\alpha = 0.3$ here). 
(c) The area–holonomy relationship is linear, $\Delta z = 2\alpha \cdot \text{Area}(\gamma)$. 
(d) Parallel transport results in a uniform rotation of tangent vectors. 
(e) The energy–connection trade‑off: increasing $\alpha$ yields more memory at a higher energetic cost. 
(f) The 2D projection shows a characteristic spiral pattern. 
(g) Memory accumulates linearly with each loop, adding $\Delta z = 2\pi\alpha R^2$ for a circle of radius $R$. 
System classification: geometrically nonconservative (fibration with a curved connection).}
\label{fig:helical}
\end{figure}

\subsection{Cylindrical Fibration: non‑simply‑connected base without holonomy}
\label{subsec:cylindrical}

Path‑dependent memory requires curvature, not merely topological intricacy.

\paragraph{Bundle structure.}
Let
\[
P \;=\; \bigl(\mathbb R^{2}\setminus\{0\}\bigr)\times S^{1},
\qquad
C \;=\; \mathbb R^{2}\setminus\{0\},
\qquad
\Phi\bigl(x,y,\vartheta\bigr)=(x,y).
\]
The fibre is $S^{1}$ (angle coordinate~$\vartheta$). The total space is a trivial product bundle.

\paragraph{Flat connection}
Set $\omega=d\vartheta$ on $P=(\mathbb R^2\!\setminus\!\{0\})\times S^1$.
Then $F=d\omega=0$ globally. Parallel transport yields $\dot\vartheta\equiv 0$,
so $\Delta\vartheta=0$ for all closed loops, contractible or not.

\paragraph{Holonomy and energetics.}
For a closed cognitive loop, horizontality imposes \(\dot\vartheta(t)=0\).
Hence $\Delta\vartheta=0$ for all closed loops $\gamma$, independent of winding number. Because horizontal lifts never move in the fibre direction, metric‑lift cost equals the geometric minimum.

\paragraph{Classification.}
The cylindrical fibration is therefore conditionally conservative: a non‑simply‑connected base space with a flat connection that stores no path history.

\begin{figure}[H]
  \centering
  \includegraphics[width=\textwidth]{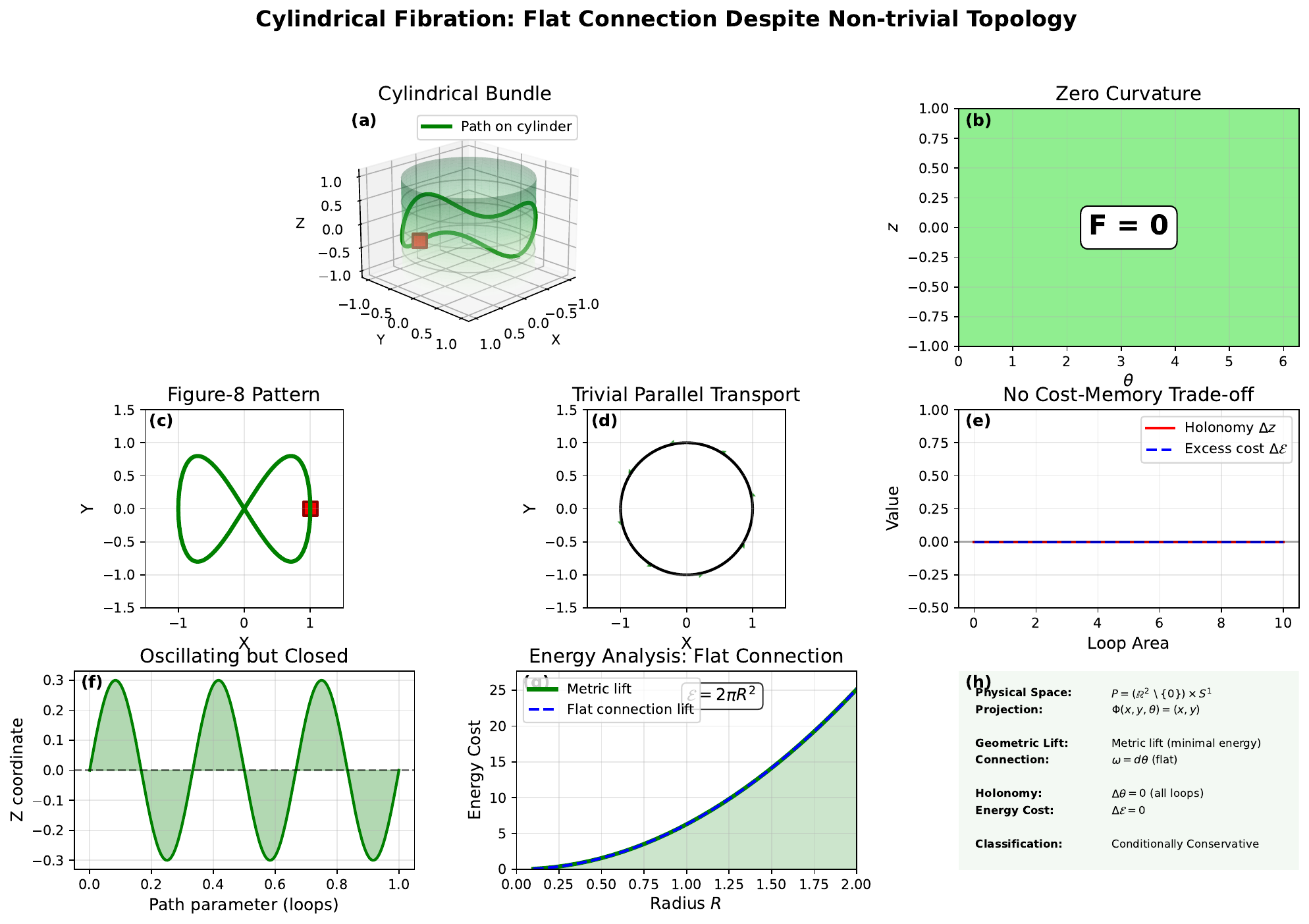}
  \caption{\textbf{Cylindrical fibration with flat connection.}
  (a) Product bundle $(\mathbb R^{2}\setminus\{0\})\times S^{1}$.
  (b) Curvature $F=0$ everywhere.
  (c) A figure‑8 cognitive loop and a loop encircling the
      origin both lift to closed paths in $P$.
  (d) Fibre angle \(\vartheta(t)\) remains constant, so holonomy
      $\Delta\vartheta=0$.
  (e) Metric‑lift energy equals the theoretical minimum; there is no
      cost–memory trade‑off.}
  \label{fig:cylindrical}
\end{figure}

\subsection{Twisted Fibration: hybrid curvature and nonlinear memory}

To model more complex, spatially varying memory effects, consider the fibration $P = \mathbb{R}^3(x,y,z) \to C = \mathbb{R}^2(x,y)$, with connection one‑form
\begin{equation}
\omega = dz - (\alpha + \beta\cos\theta)(x\,dy - y\,dx),
\end{equation}
where $\theta = \arctan(y/x)$, $\alpha$ is a constant drift and $\beta$ a variable twist. In polar coordinates $(r,\theta)$, $\omega = dz - r^2(\alpha + \beta\cos\theta)\,d\theta$.
The curvature is
\begin{equation}
F = d\omega = 2(\alpha + \beta\cos\theta) \,dx \wedge dy .
\end{equation}
For a circular trajectory centred at the origin, the $\beta$ term integrates to zero and the net holonomy is $\Delta z = 2\pi\alpha R^2$. For off‑centre paths, the $\beta$ term contributes, creating position‑dependent memory.

\begin{figure}[H]
\centering
\includegraphics[width=\textwidth]{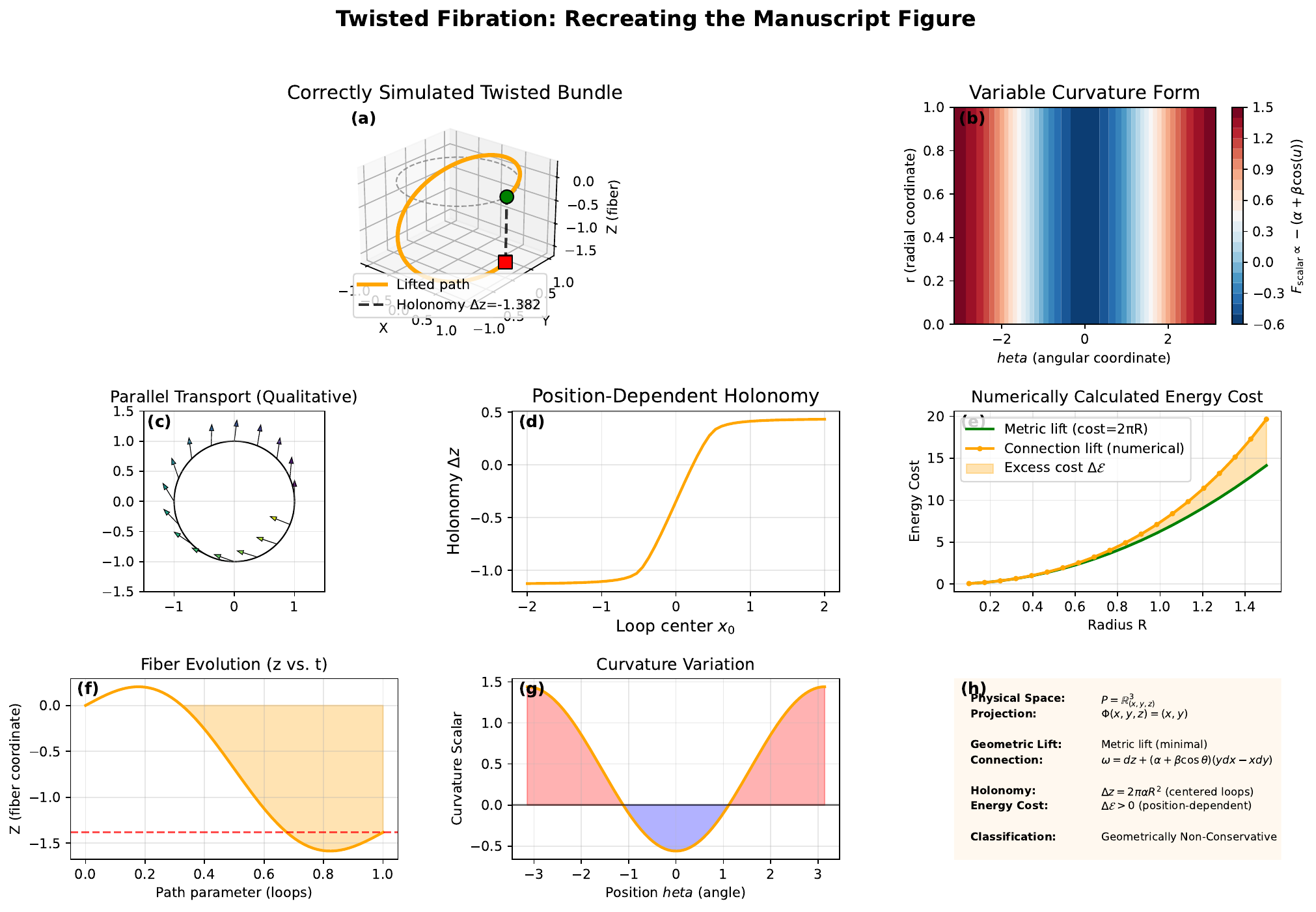}
\caption{\textbf{Twisted fibration: hybrid curvature creates nonlinear memory.}
(a) Twisted bundle structure with net vertical displacement $\Delta z = -1.382$.
(b) Variable curvature $F \propto (\alpha + \beta\cos \theta)$ creates position dependence.
(c) Parallel transport does not return vectors to their initial orientation.
(d) Off‑centre loops show complex holonomy landscapes.
(e) Energy increases nonlinearly compared to the metric lift.
(f) Fibre coordinate shows a linear drift from $\alpha$ plus sinusoidal modulation from $\beta$.
(g) Regions of stronger and weaker curvature alternate with $\theta$.
System classification: geometrically nonconservative (curved connection).}
\label{fig:twisted}
\end{figure}

\section{Integrated Intentional Dynamics: cost-aware goal selection}
\label{sec:intentional}

This section couples the energetic picture on $(P,G)$ to goal evolution on $I$ through the $P\!\to\!C\!\to\!I$ chain. The aim is to show, in a compact and type‑correct way, how local travel cost in $P$ biases goals in $I$ and how this plays out in the strip–sine example.

\subsection{The \texorpdfstring{$P\text{–}C\text{–}I$}{P–C–I} chain with energetic feedback}
At the physical level, trajectories pay the squared‑speed cost
$\mathcal{E}[\gamma]=\int\|\dot u\|_{G}^{2}\,dt$. Among all lifts, the metric lift is instantaneously optimal, and any memory‑bearing deviation (vertical motion) incurs a strictly positive excess; this is the cost–memory duality developed in §§3.1–3.2. 
To transmit this information “upwards,” let $s=\Psi(c)\in I$ and define the push‑forward operator on $T_sI$
\[
M(c):=D\Psi_c\,g_C^{-1}(D\Psi_c)^{\top}g_I,
\]
which is symmetric positive‑definite under the standing rank assumptions. The cost‑aware intentional dynamics we standardize on (natural‑gradient preconditioning on $I$) are
\begin{equation}\label{eq:intentional_dynamics_pullback}
  \dot s \;=\; -\,\mu_I\,M(c)^{-1}\,\nabla_I V(s)\;-\;\eta_P\,D\Psi_c\,\nabla_C E_{\mathrm{loc}}(c),
\end{equation}
with gains $\mu_I,\eta_P>0$ and gradients taken with respect to $g_I,g_C$ (so both terms lie in $T_sI$). The first term descends the intentional potential along the natural gradient on $I$; the second term pushes back local cognitive cost into $I$ via $D\Psi_c:T_cC\to T_sI$. 

To make the $C$–term a bona‑fide gradient, we use the metric‑lift surrogate
\[
M_C(c):=D\Phi\,G^{-1}D\Phi^{\top},
\qquad
E_{\mathrm{loc}}(c):=\mathbb{E}_{\dot c\sim\pi(\cdot\,|\,c)}\big[\dot c^{\top}M_C(c)^{-1}\dot c\big],
\]
so that $\nabla_C E_{\mathrm{loc}}(c)$ points toward cognitively expensive regions; a practical analytic proxy is $E_{\mathrm{loc}}(c)=\operatorname{tr}\big(M_C(c)^{-1}\big)$ (uniform directional average). 

\subsection{Strip–sine: emergence of lazy but effective behaviour}
For the strip–sine system of §6.1 with $\Phi(u,v)=(u,v+\kappa\sin u)$, take $g_C=I_2$, $g_I=1$, $G=I_2$, $\Psi(c)=s=c_1^2+c_2^2$, and $V(s)=\tfrac12(s-s_\star)^2$. Then
\[
D\Psi_c=[\,2c_1\;\;2c_2\,],\qquad M(c)=D\Psi_cD\Psi_c^\top=4(c_1^2+c_2^2)=4s,
\]
and, with $D\Phi=\begin{bmatrix}1&0\\ \kappa\cos c_1&1\end{bmatrix}$,
\[
M_C(c)^{-1}=
\begin{bmatrix}
1 & -\,\kappa\cos c_1\\[2pt]
-\,\kappa\cos c_1 & 1+\kappa^2\cos^2 c_1
\end{bmatrix},
\quad
E_{\mathrm{loc}}(c)=\operatorname{tr}M_C^{-1}=2+\kappa^2\cos^2 c_1,
\quad
\nabla_C E_{\mathrm{loc}}=(-\kappa^2\sin(2c_1),\,0).
\]
Substitution into \eqref{eq:intentional_dynamics_pullback} yields the scalar ODE
\[
\dot s
=\,-\,\mu_I\,\frac{1}{4s}\,(s-s_\star)\;+\;2\,\eta_P\,\kappa^{2}\,c_1\,\sin(2c_1),
\]
which makes two effects explicit: (i) \emph{natural‑gradient scaling} of goal descent by $1/(4s)$; (ii) \emph{cost pushback} that steers $s$ away from $c_1$–regions with large $|\cos c_1|$. The area–holonomy and cost scalings used below come directly from §6.1’s prescribed‑dynamics analysis. 

For a circular cognitive loop of radius $R$ one has $\Delta h=2\alpha\pi R^2$ and the excess energy
$\Delta\mathcal{E}=2\pi\alpha^2R^4=(\Delta h)^2/(2\pi)$; see Fig.\ref{fig:strip_sine} for the linear area–holonomy law and the quadratic cost–memory relation. Hence, for a fixed total holonomy $H=\sum_{i=1}^m\Delta h_i$ accumulated over $m$ loops,
\[
\sum_{i=1}^m\Delta\mathcal{E}_i
=\frac{1}{2\pi}\sum_{i=1}^m(\Delta h_i)^2
\;\ge\;\frac{H^2}{2\pi\,m},
\]
with equality when the $\Delta h_i$ are equal (Cauchy–Schwarz). The dynamics \eqref{eq:intentional_dynamics_pullback} therefore favour decomposing one large loop into several smaller ones when exploration is useful but energy is penalised, while the $-\eta_P D\Psi_c\nabla_C E_{\mathrm{loc}}$ term suppresses holonomy when memory provides no benefit.

\section{Reflective TAS (rTAS): self‑referential agents}\label{sec:rTAS}

\paragraph{Motivation.}
Living systems also carry internal models that co‑evolve with their interaction dynamics. Temporalising self‑reference resolves apparent paradoxes and motivates augmenting TAS with a model manifold so that perception–action and model–updating are treated in one geometric object. In rTAS, agents can trade physical effort against model change while retaining the TAS energy–holonomy logic. For a recent synthesis on time, self‑reference and self‑modification across biology and computation, see \cite{abramsky2025openquestionstimeselfreference}, where the authors distinguish natural time from representational time and show how self‑reference is unwound by explicit temporal structure.
\emph{Visual overview.}  The single‑agent reflective mechanism—instantaneous effort split, the $\lambda$ trade‑off, and the TAS limit—is summarised in \cref{fig:rtas1}.  How “re‑entry’’ is implemented geometrically via cross‑curvatures is shown in \cref{fig:rtas2}.  The energetic consequences and the empirical cost–memory frontier appear in \cref{fig:rtas3}.  Two minimal working channels used throughout (projective and connection) are documented in \cref{fig:rtasA,fig:rtasB}.  Brief two‑agent illustrations (for self‑reference at the population level) are collected in \cref{fig:ma_exp1,fig:ma_exp2,fig:ma_exp3,fig:ma_exp4}.

\subsection{Construction: adding a model manifold and reflective projections}

\begin{definition}[Reflective TAS]\label{def:rTAS}
Let $(P,G)$, $C$, $I$, and the submersions $\Phi:P\!\to\!C$, $\Psi:C\!\to\!I$ be as in \cref{def:TAS}. Let $(M,H)$ be a finite‑dimensional Riemannian model manifold. Define the extended spaces
\[
\widehat P := P\times M,\qquad \widehat C := C\times M,\qquad \widehat I := I\times M.
\]
A \emph{reflective TAS (rTAS)} specifies a smooth family of projections $\{\Phi_m:P\!\to\!C,\ \Psi_m:C\!\to\!I\}_{m\in M}$ that assemble into
\[
\widehat\Phi(p,m)=\big(\Phi_m(p),\,m\big),\qquad
\widehat\Psi(c,m)=\big(\Psi_m(c),\,m\big),
\]
with the rank condition
\[
\operatorname{rank}\!\left[\,D\Phi_m(p)\ \ \ \partial_m\Phi_m(p)\,\right]=\dim C\quad\text{for all }(p,m).
\]
\label{def:rtas}
\end{definition}

\begin{proposition}[Consistency]\label{prop:rTASreduces}
If $M$ is a singleton or $\partial_m\Phi_m\equiv 0$, then $\Phi_m$ is independent of $m$ and
$\widehat\Phi(p,m)=(\Phi(p),m)$; hence rTAS reduces to TAS.
\end{proposition}

\subsection{Reflective metric lift and instantaneous optimality}
With the block metric $\widehat G=G\oplus\lambda H$, the reflective constraint
\[
D\Phi_m(p)\,\dot u + \partial_m\Phi_m(p)\,\dot m=\dot c
\]
forces the instantaneous trade between physical velocity $\dot u$ and model velocity $\dot m$.  The closed‑form solution in \cref{def:reflective_lift} allocates effort between these channels (\cref{fig:rtas1}\,b) and interpolates from flexible (small~$\lambda$) to frozen model (large~$\lambda$), see panels~(a,c–e) of \cref{fig:rtas1}.

Equip $\widehat P$ with the block metric $\widehat G:= G\oplus \lambda H$ with $\lambda>0$. Given $(p,m)\in\widehat P$ and a desired cognitive velocity $\dot c\in T_{\Phi_m(p)}C$, rTAS imposes the reflective constraint
\begin{equation}\label{eq:reflective_constraint}
D\Phi_m(p)\,\dot u\;+\;B_{(p,m)}\,\dot m \;=\; \dot c, \qquad B_{(p,m)}:=\partial_m\Phi_m(p).
\end{equation}

\begin{definition}[Reflective metric lift]\label{def:reflective_lift}
Let $A_{(p,m)}:=[\,D\Phi_m(p)\ \ B_{(p,m)}\,]$ and $\widehat G^{-1}=\operatorname{diag}(G^{-1},\,\lambda^{-1}H^{-1})$. The reflective metric lift is
\begin{equation}\label{eq:block_lift_formula}
\begin{bmatrix}\dot u^\star\\ \dot m^\star\end{bmatrix}
=
\widehat G^{-1}A_{(p,m)}^{\!\top}
\Big(A_{(p,m)}\,\widehat G^{-1}A_{(p,m)}^{\!\top}\Big)^{-1}\dot c,
\end{equation}
which uniquely minimises $\|\dot u\|_{G}^{2}+\lambda\|\dot m\|_{H}^{2}$ subject to \eqref{eq:reflective_constraint}.
\end{definition}

\begin{lemma}[Instantaneous optimality]\label{lem:instant_opt}
Under the rank assumption in \cref{def:rTAS}, the solution \eqref{eq:block_lift_formula} exists and is unique. Moreover, for any feasible $(\dot u,\dot m)$,
\[
\|\dot u^\star\|_{G}^{2}+\lambda\|\dot m^\star\|_{H}^{2}
\ \le\
\|\dot u\|_{G}^{2}+\lambda\|\dot m\|_{H}^{2},
\]
with equality iff $(\dot u,\dot m)=(\dot u^\star,\dot m^\star)$.
\end{lemma}

\begin{remark}[Effort–learning trade]
The cross term $B_{(p,m)}$ allows part of $\dot c$ to be realised by model motion $\dot m$, explicitly trading physical effort for model change. If $B\equiv 0$, then $\dot m^\star=0$ and the TAS metric lift is recovered.
\end{remark}

\begin{figure}[htbp]
\centering
\includegraphics[width=\textwidth]{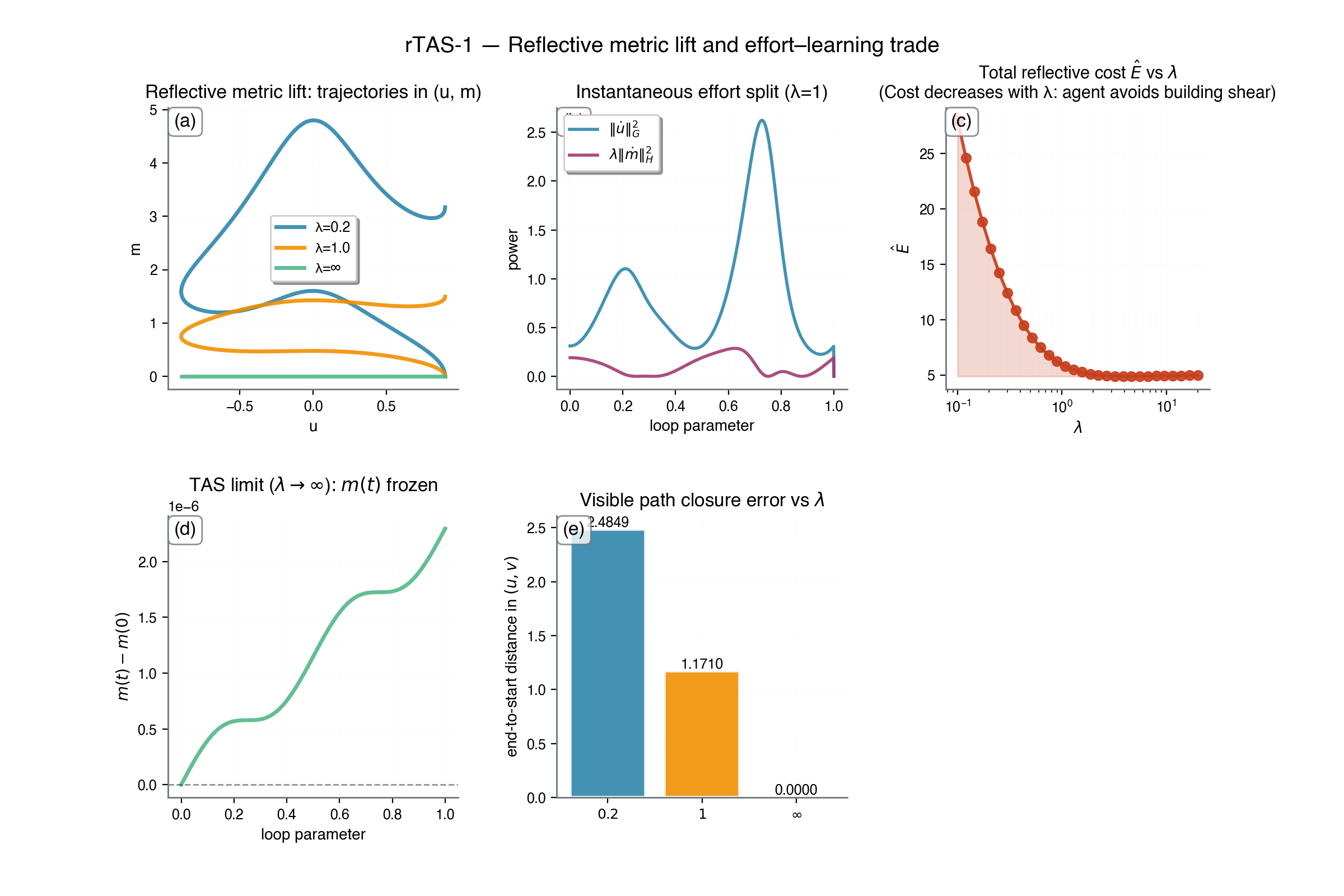}
\caption{\textbf{Reflective metric lift and effort–learning trade.}
Panels (a–e) illustrate the reflective lift in \cref{def:reflective_lift} (cost in \cref{def:rCost}).
\textbf{(a)} Trajectories in $(u,m)$ for $\lambda\!\in\!\{0.2,1,\infty\}$: as $\lambda$ grows the model motion is suppressed and the TAS limit is approached.
\textbf{(b)} Instantaneous effort split for $\lambda{=}1$: physical power $\|\dot u\|_G^2{+}\|\dot v\|_G^2$ (blue) and model power $\lambda\|\dot m\|_H^2$ (maroon) are allocated by the reflective lift.
\textbf{(c)} Total reflective cost $\hat E$ versus $\lambda$ (log abscissa): here $\hat E$ \emph{decreases} and saturates because penalising model motion suppresses the build‑up of shear, reducing activity in both channels for this task/channel. For a detailed discussion see Sec.( \ref{subsec:lambda-monotonicity}).
\textbf{(d)} TAS limit $\lambda\!\to\!\infty$: $m(t)$ is effectively frozen (numerical drift at $10^{-6}$ scale).
\textbf{(e)} Visible path closure error in $(u,v)$ versus $\lambda$: non‑closure at small~$\lambda$ disappears in the TAS limit.}

\label{fig:rtas1}
\end{figure}

\subsection{Reflective connections and block curvature}
Writing the reflective connection in block form turns \emph{re‑entry} into geometry: $F_{pm}$ and $F_{mp}$ quantify how model motion induces physical holonomy and vice versa, cf.\ \cref{fig:rtas2}\,(a–c,e).  The block diagram in panel~(d) summarises the four curvature channels.

Let $\omega_m$ be a connection one‑form for $\Phi$ depending smoothly on $m$, and let $\Gamma_{(p,m)}:T_mM\!\to$ fibre coordinates specify how model motion affects horizontality. A reflective horizontal lift at $(p,m)$ solves
\[
\min_{\dot u,\dot m}\ \|\dot u\|_{G}^{2}+\lambda\|\dot m\|_{H}^{2}
\quad \text{s.t.}\quad
\begin{cases}
D\Phi_m(p)\,\dot u + B_{(p,m)}\,\dot m = \dot c,\\
\omega_m(\dot u) + \Gamma_{(p,m)}\,\dot m = 0.
\end{cases}
\]
This yields a block connection on $\widehat\Phi$ with curvature
\[
\widehat F
= d\widehat\omega+\widehat\omega\wedge\widehat\omega
=
\begin{pmatrix}
F_{pp} & F_{pm}\\
F_{mp} & F_{mm}
\end{pmatrix}.
\]
\paragraph{Cross-curvatures.}
Writing the reflective connection one-form in block form
$\widehat\omega=\begin{psmallmatrix}\omega_{pp}&\omega_{pm}\\ \omega_{mp}&\omega_{mm}\end{psmallmatrix}$,
the curvature $\widehat F=d\widehat\omega+\widehat\omega\wedge\widehat\omega$ has components
$F_{pp},F_{pm},F_{mp},F_{mm}$. Here $F_{pp}$ is the usual physical curvature (geometric memory);
$F_{mm}$ is model-space curvature (meta-memory); $F_{pm}$ measures how model variations twist
the physical horizontality (model change \emph{induces} physical holonomy); $F_{mp}$ measures the
reciprocal effect (physical motion \emph{induces} meta-holonomy). The projective-channel example
makes $F_{pm}\neq 0$ explicit.

\begin{proposition}[Holonomy decomposition]\label{prop:block_hol}
Let $\gamma\subset C$ be a small closed loop and $\widehat\gamma$ its reflective horizontal lift. For Abelian fibres,
\[
\Delta u_{\mathrm{phys}}=\iint_{S_\gamma}\!F_{pp}+\iint_{S_\gamma}\!F_{pm},
\qquad
\Delta m_{\mathrm{model}}=\iint_{S_\gamma}\!F_{mp}+\iint_{S_\gamma}\!F_{mm}.
\]
Thus cross‑curvatures implement re‑entry: model changes can induce physical holonomy and vice versa.
\end{proposition}

\begin{figure}[htbp]
\centering
\includegraphics[width=\textwidth]{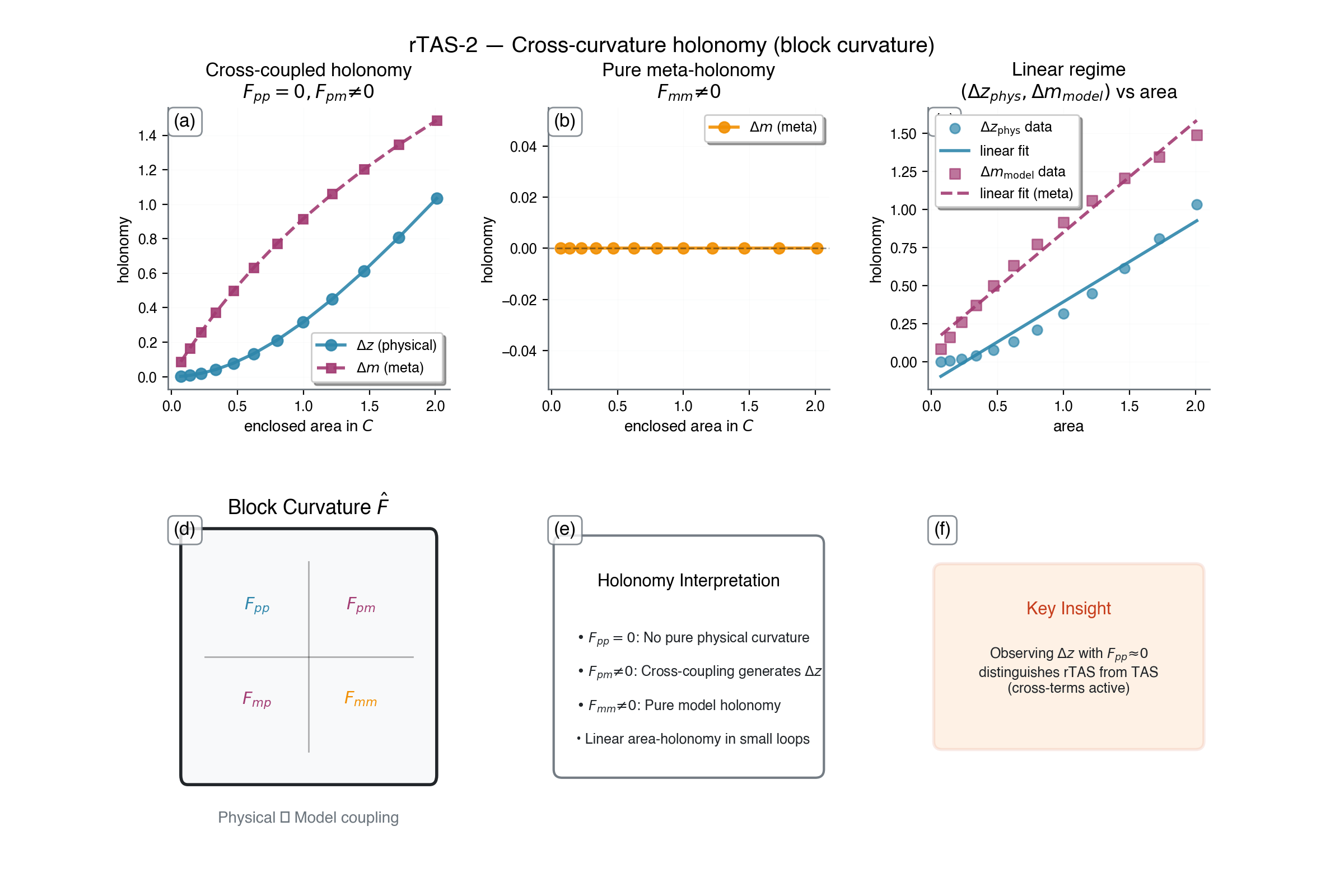}
\caption{\textbf{Cross‑curvature holonomy (block curvature).}
\textbf{(a)} Cross‑coupled holonomy with $F_{pp}=0$ and $F_{pm}\neq 0$: despite vanishing pure physical curvature, the cross‑term generates visible displacement $\Delta z$ (blue) alongside meta‑holonomy $\Delta m$ (magenta), both growing with enclosed area in~$C$ (linear for small loops).
\textbf{(b)} Pure meta‑holonomy with $F_{mm}\neq 0$: $\Delta m$ scales linearly with area while $\Delta z\simeq 0$.
\textbf{(c)} Linear regime: $(\Delta z_{\mathrm{phys}},\Delta m_{\mathrm{model}})$ versus area with least‑squares fits.
\textbf{(d)} Block curvature matrix $\widehat F=\begin{psmallmatrix}F_{pp}&F_{pm}\\ F_{mp}&F_{mm}\end{psmallmatrix}$.
\textbf{(e)} Interpretation of the four channels (pure physical, pure model, and the two cross‑curvatures implementing re‑entry).
\textbf{(f)} Diagnostic: observing $\Delta z$ with $F_{pp}\approx 0$ identifies reflective coupling (rTAS) rather than ordinary TAS.}

\label{fig:rtas2}
\end{figure}

\subsection{Energetics and the extended cost–memory law}
rTAS extends the TAS cost law by pricing \emph{both} physical holonomy and meta‑holonomy.  In small loops the excess cost is quadratic in the combined holonomy vector, as visualised empirically in \cref{fig:rtas3}\,(a).  Panels (b–c) illustrate how cost feedback in $I$ produces “lazy but effective’’ behaviour and how coverage can be achieved at lower reflective cost.

\begin{definition}[Reflective travel cost]\label{def:rCost}
For a lifted trajectory $\widehat\gamma(t)=(p(t),m(t))$ define
\[
\widehat{\mathcal E}[\widehat\gamma]\;=\;\int \Big(\|\dot u(t)\|_{G}^{2}+\lambda\|\dot m(t)\|_{H}^{2}\Big)\,dt.
\]
\end{definition}

\begin{theorem}[Extended cost–memory duality]\label{thm:extended_tradeoff}
Fix a closed cognitive loop $\gamma$. Among all admissible reflective lifts, the reflective metric lift minimises $\widehat{\mathcal E}$. Moreover, in the small‑loop regime there exists a positive‑definite quadratic form $Q$ such that
\[
\widehat{\mathcal E}[\widehat\gamma]-\widehat{\mathcal E}_{\min}[\gamma]
\ \ge\
\begin{bmatrix}\Delta u_{\mathrm{phys}}\\ \Delta m_{\mathrm{model}}\end{bmatrix}^{\!\top}
Q
\begin{bmatrix}\Delta u_{\mathrm{phys}}\\ \Delta m_{\mathrm{model}}\end{bmatrix}
\ +\ o\big(\operatorname{Area}(\gamma)^{2}\big).
\]
Hence excess energy scales quadratically in both physical holonomy and meta‑holonomy.
\end{theorem}

\begin{corollary}[Useful limits]\label{cor:rLimits}
(a) $\lambda\!\to\!\infty$ freezes $m$ and recovers TAS. \;
(b) If $B\!=\!\Gamma\!=\!0$ then $F_{pm}\!=\!F_{mp}\!=\!F_{mm}\!=\!0$ and no meta‑holonomy arises. \;
(c) If $F_{pp}\!=\!0$ but $F_{mp}$ or $F_{mm}\!\neq\!0$, memory is stored purely in $M$.
\end{corollary}

\begin{figure}[htbp]
\centering
\includegraphics[width=\textwidth]{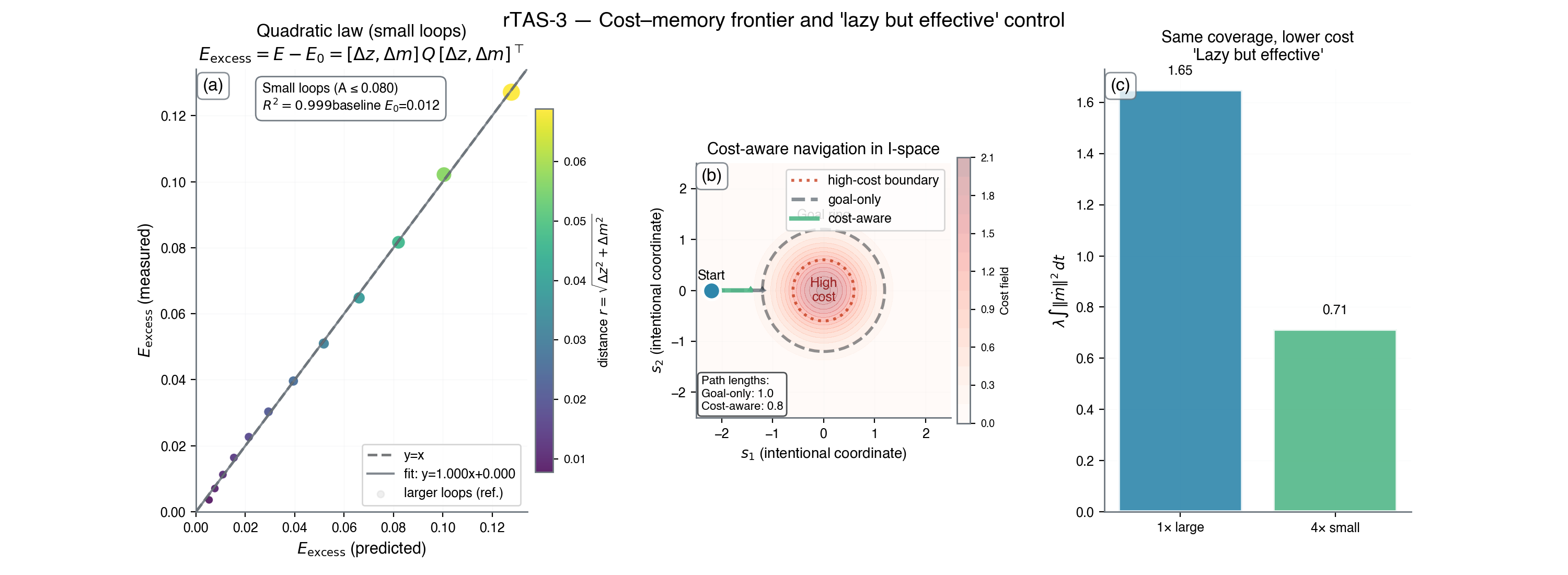}
\caption{\textbf{Cost–memory frontier and lazy but effective control.}
\textbf{(a)} Small–loop quadratic law using a zero–holonomy baseline. The excess reflective cost
\(E_{\text{excess}} := E - E_0\) (with \(E_0\) the rTAS energy at \(\Delta z=\Delta m=0\))
obeys \(E_{\text{excess}} \approx [\Delta z,\Delta m]\;Q\;[\Delta z,\Delta m]^{\top}\).
Scatter shows measured vs.\ predicted \(E_{\text{excess}}\) for small loops \((A \le 0.080)\);
colour encodes \(r=\sqrt{\Delta z^{2}+\Delta m^{2}}\). The fit is essentially on the identity
(\(R^{2}=0.999\), baseline \(E_{0}=0.012\)).
\textbf{(b)} Cost‑aware navigation in intentional space \(I\): adding a push‑back term
proportional to \(-\eta_{P}\nabla_{C}E_{\text{loc}}\) bends the trajectory around the high‑cost
region (shaded/contours), yielding a shorter path than goal‑only dynamics (legend).
\textbf{(c)} \emph{Lazy but effective.} Decomposing one large loop into four smaller loops achieves
the same spatial coverage with substantially lower model‑activity cost
\(\lambda\!\int\!\|\dot m\|^{2}\,dt\) (bars; \(1\times\) large \(\approx 1.65\) vs.\ \(4\times\) small \(\approx 0.71\)).}

\label{fig:rtas3}
\end{figure}

\subsection{Minimal worked examples}

\paragraph{Projective channel (reflective strip–sine).}
Let $P=\mathbb R^2_{(u,v)}$, $C=\mathbb R^2_{(c_1,c_2)}$, and internalize the shear magnitude by
\[
\Phi_m(u,v)=\big(u,\ v+m\sin u\big),
\quad
D\Phi_m=\begin{psmallmatrix}1&0\\ m\cos u&1\end{psmallmatrix},\ \
B=\partial_m\Phi_m=\begin{psmallmatrix}0\\ \sin u\end{psmallmatrix}.
\]
With $G=\mathbf I_2$, $H=1$, the reflective metric lift yields
\[
\dot u^\star=\dot c_1,\qquad
\begin{bmatrix}\dot v^\star\\ \dot m^\star\end{bmatrix}
=
\frac{1}{\lambda+\sin^2 u}
\begin{bmatrix}\lambda\\ \sin u\end{bmatrix}
\Big(\dot c_2 - m\cos u\,\dot c_1\Big),
\]
so closed loops in $C$ generate generically $\Delta m_{\mathrm{model}}\!\neq\!0$ while the visible path closes; see \cref{fig:rtasA} for the full behaviour across~$\lambda$ (panels a–e).

\paragraph{Connection channel (reflective helical or twisted fibrations).}
Let $\Phi(x,y,z)=(x,y)$ as before, but index the connection by $m$:
\[
\omega_m=dz-\alpha(m)\,(y\,dx-x\,dy)
\quad\text{or}\quad
\omega_m=dz-\big(\alpha(m)+\beta(m)\cos\theta\big)\,(x\,dy-y\,dx).
\]
Then $F_{pp}$ produces the familiar area‑law holonomy in $z$, while cross‑curvatures govern meta‑holonomy in $m$, tunable via $\lambda$; see \cref{fig:rtasB}\,(a–d).

\begin{figure}[htbp]
\centering
\includegraphics[width=\textwidth]{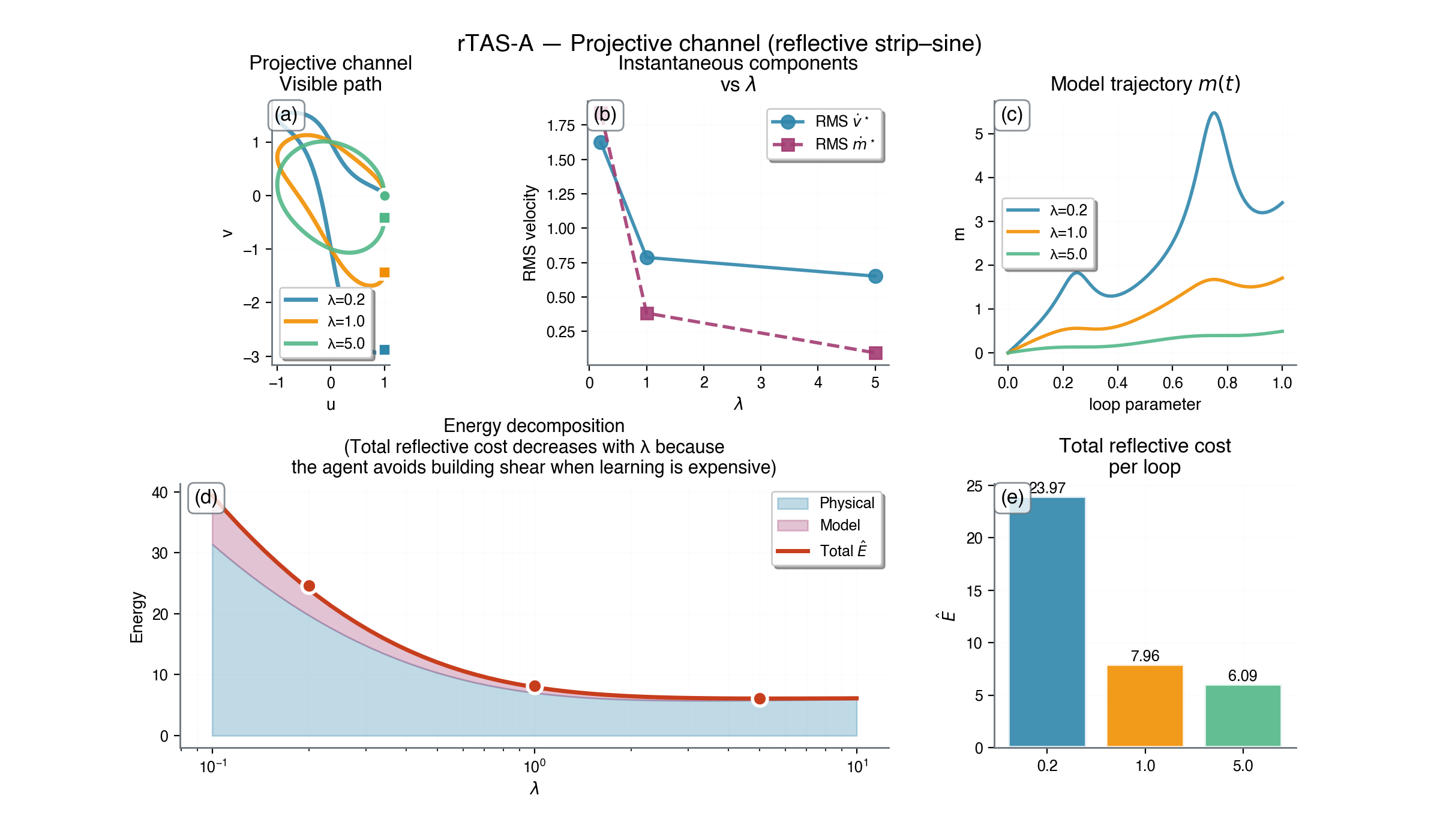}
\caption{\textbf{Projective channel (reflective strip–sine).}
\textbf{(a)} Visible paths $(u(t),v(t))$ for $\lambda\!\in\!\{0.2,1,5\}$ (start $\circ$, end $\square$); closure improves with~$\lambda$.
\textbf{(b)} Instantaneous components (RMS over the loop): both $\mathrm{RMS}(\dot v^\star)$ and $\mathrm{RMS}(\dot m^\star)$ decrease as $\lambda$ increases—suppressing $m$ shrinks the mixing term $w=\dot c_2-m\cos u\,\dot c_1$ and calms both channels.
\textbf{(c)} Model trajectories $m(t)$ shrink as $\lambda$ grows, tending to the TAS limit.
\textbf{(d)} Energy decomposition versus $\lambda$ (log abscissa): physical and model contributions both drop as the system avoids building shear when learning is expensive; the total $\hat E$ decreases and saturates.
\textbf{(e)} Per‑loop total reflective cost $\hat E$ for the three~$\lambda$’s.}

\label{fig:rtasA}
\end{figure}

\begin{figure}[htbp]
\centering
\includegraphics[width=\textwidth]{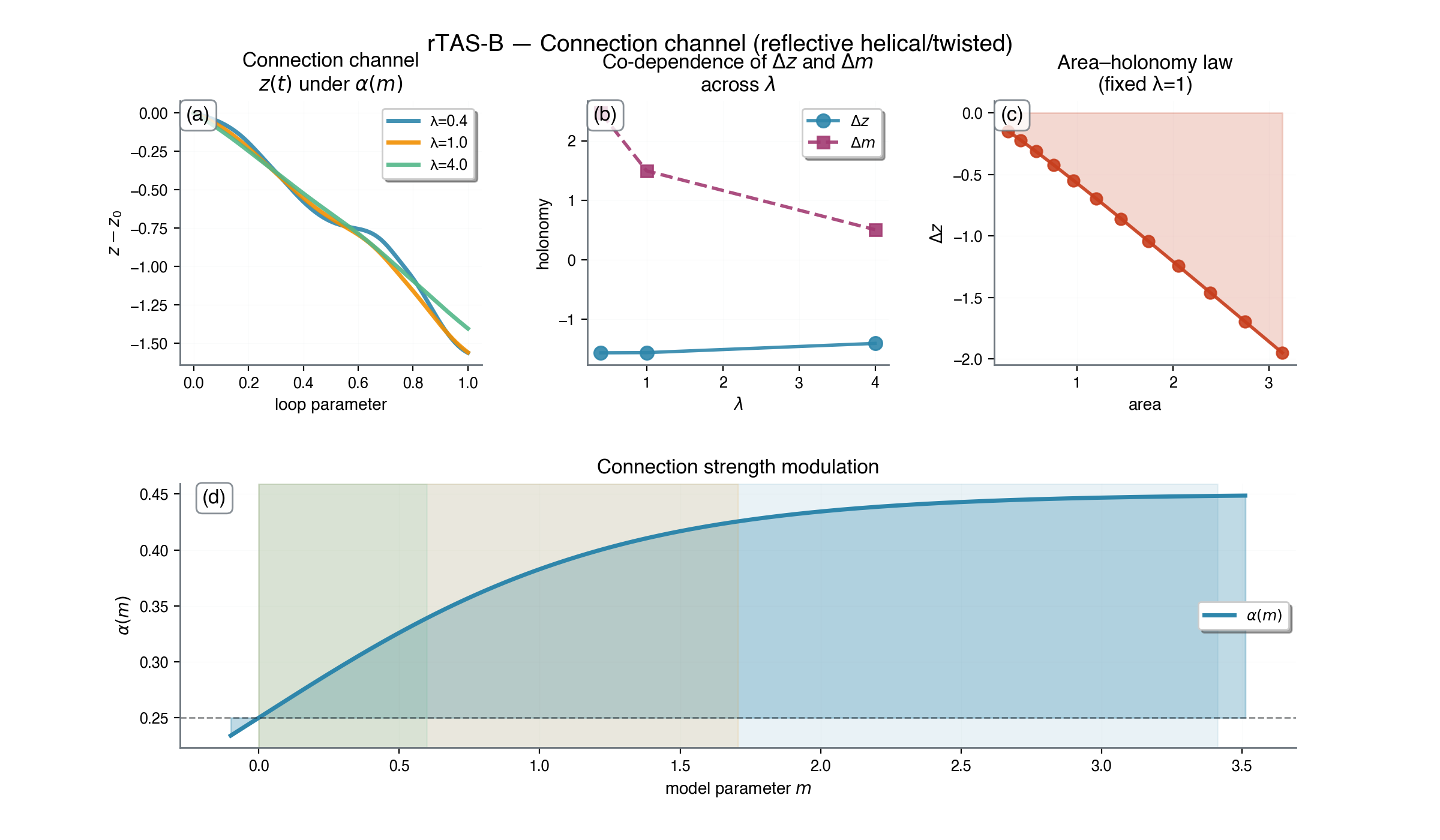}
\caption{\textbf{Connection channel (reflective helical/twisted).}
\textbf{(a)} Lifted $z(t)$ under a model‑modulated connection $\omega_m = dz-\alpha(m)(y\,dx-x\,dy)$ for several $\lambda$.
\textbf{(b)} Co‑dependence of $\Delta z$ and $\Delta m$ across $\lambda$: $\Delta m$ decreases as learning becomes expensive, while $\Delta z$ varies only weakly in this construction.
\textbf{(c)} Area–holonomy law at fixed $\lambda{=}1$: $\Delta z$ is (approximately) linear in the enclosed area (sign reflects loop orientation).
\textbf{(d)} Connection strength modulation $\alpha(m)$ controlling coupling; shaded $m$‑ranges indicate the values explored by the trajectories in (a–c).}

\label{fig:rtasB}
\end{figure}

\subsection{Coupled reflective agents (two‑agent illustrations)}

\paragraph{Set‑up and spaces.}
We consider two reflective agents $i\in\{1,2\}$ that share a perceptual channel but carry their own internal models.  For each agent
\[
P_i=\R^2_{(u_i^1,u_i^2)},\qquad C_i=\R^2_{(c_i^1,c_i^2)},\qquad M_i=\R,\qquad I_i\ \text{(implicit, goal field in $C_i$)},
\]
so that $\dim P_i=\dim C_i=2$, $\dim M_i=1$.\footnote{In these experiments $I_i$ is not modelled explicitly; goals are specified as vector fields $\dot c_i$ in $C_i$ (e.g.\ formation tracking or pursuit). This is equivalent to taking $\Psi_{m_i}=\mathrm{id}_{C_i}$ or to prescribing an intentional potential on $C_i$.}
The joint stacked state is
\[
\widehat x=\Big((u_1,m_1),(u_2,m_2)\Big)\in\widehat P_1\times\widehat P_2,\quad
\widehat P_i:=P_i\times M_i,
\]
with joint cognitive target $(\dot c_1,\dot c_2)\in T_{c_1}C_1\times T_{c_2}C_2$.
Each agent uses the same \emph{projective channel} 
\[
\Phi_{m_i}(u_i^1,u_i^2)=\big(u_i^1,\ u_i^2+m_i\sin u_i^1\big),
\quad
D\Phi_{m_i}=
\begin{bmatrix}
1&0\\ m_i\cos u_i^1&1
\end{bmatrix},\ 
B(u_i):=\partial_{m}\Phi_{m_i}=
\begin{bmatrix}
0\\ \sin u_i^1
\end{bmatrix},
\]
and block metric $\widehat G_i:=G\oplus\lambda_i H$ with $G=\mathbf I_2$, $H=1$.

\paragraph{Coupling channel.}
Agents are coupled \emph{through perception}: the other agent’s current model state $m_j$ ($j\neq i$) distorts agent $i$’s effective cognitive velocity by an additive field $g(u_i)$,
\begin{equation}
\dot c_{i,\mathrm{eff}}
\;=\;
\dot c_i\;-\;\kappa\,m_j\,g(u_i),
\qquad
g(u):=\begin{bmatrix}0\\ \cos u^1\end{bmatrix},
\label{eq:coupling}
\end{equation}
where $\kappa\!\ge\!0$ is the coupling strength.  Setting \(\kappa=0\) recovers independent agents; taking the perturbation only for $i=1$ models \emph{information asymmetry} (agent~1 observes/uses $m_2$, agent~2 does not), as in the pursuit–evasion example below.

\paragraph{Per‑agent reflective lift and total energy.}
Agent $i$ solves, at each instant, the reflective metric‑lift constrained least‑squares with target~$\dot c_{i,\mathrm{eff}}$:
\begin{equation}
\begin{bmatrix}\dot u_i^\star\\ \dot m_i^\star\end{bmatrix}
=\widehat G_i^{-1}A_i^{\!\top}\!\left(A_i \widehat G_i^{-1} A_i^{\!\top}\right)^{-1}\dot c_{i,\mathrm{eff}},
\qquad
A_i:=\big[D\Phi_{m_i}\ \ B(u_i)\big].
\label{eq:2agent-lift}
\end{equation}
The coupled system’s energy and (proxy) holonomy are accumulated additively:
\[
\widehat{\mathcal E}_{\mathrm{tot}}
=\sum_{i=1}^2\int\!\Big(\|\dot u_i^\star\|^2+\lambda_i(\dot m_i^\star)^2\Big)\,dt,
\qquad
\Delta\mathsf{Hol}_i
=\int\!F_{pp}(u_i)\,dt+\int\!F_{pm}(u_i)\,m_i\,dt,
\]
where in our projective example $F_{pp}=0$ and $F_{pm}\propto \cos u_i^1$ (a simple cross‑curvature proxy).  Equation~\eqref{eq:2agent-lift} makes the \emph{self‑referential} structure explicit: each agent’s internal model $m_i$ shapes its own perceptual lift via $B(u_i)$ while simultaneously shaping the other agent’s perceived task via~\eqref{eq:coupling}.

\begin{figure}[t]
  \centering
  \includegraphics[width=0.8 \textwidth]{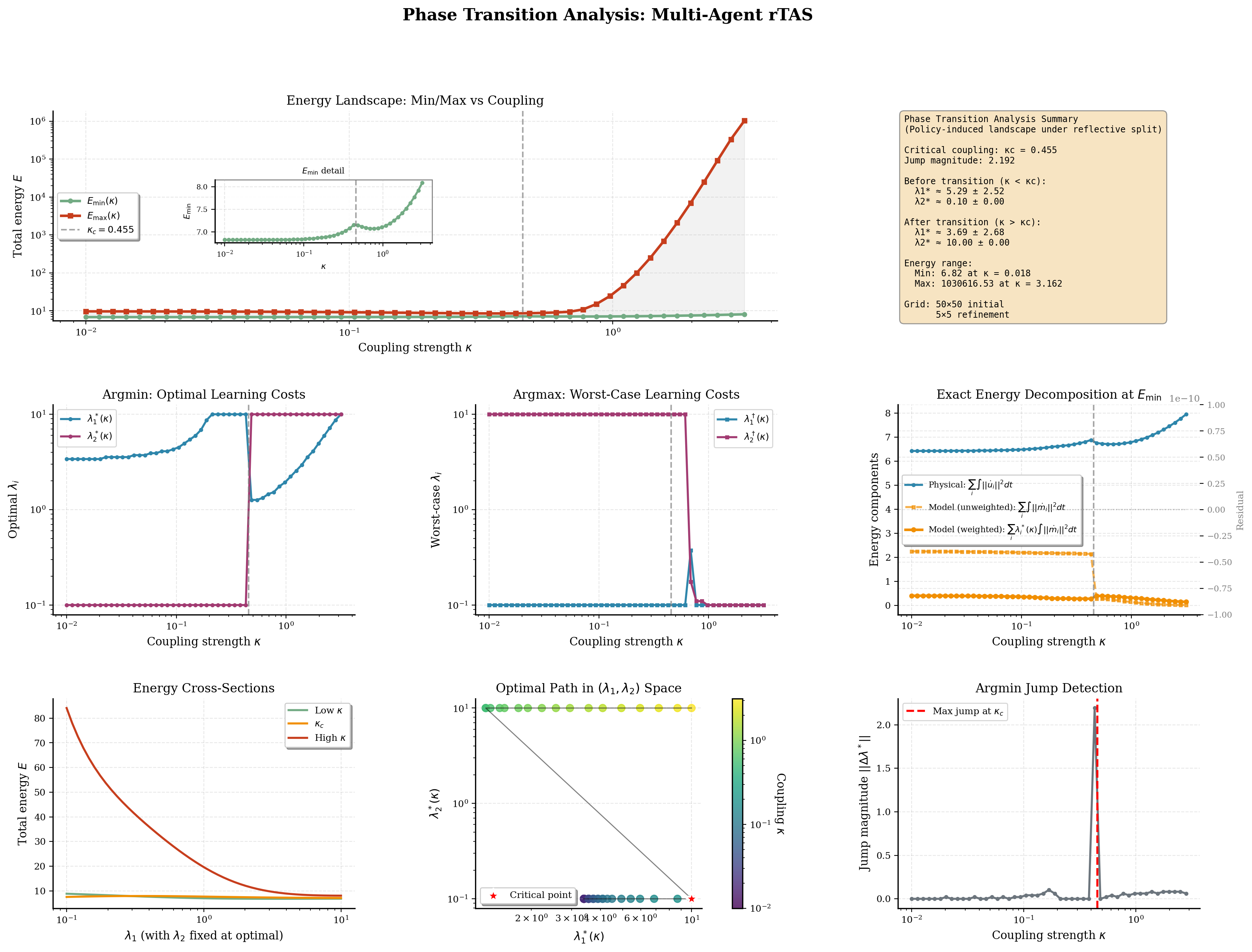}
  \caption{\textbf{Phase transition analysis for the policy–induced landscape (instantaneous reflective split).}
  \textbf{Top left:} $E_{\min}(\kappa)$ and $E_{\max}(\kappa)$ versus coupling (log $y$–axis); the dashed line marks the critical coupling $\kappa_c\!\approx\!0.455$. The inset shows the fine drift of $E_{\min}$ near $\kappa_c$.
  \textbf{Middle left:} Argmin trajectories $\lambda_i^{\*}(\kappa)$ (log–log): below $\kappa_c$ the optimum is asymmetric (leader–follower); above $\kappa_c$ both $\lambda_i^{\*}$ jump high (rigid–rigid).
  \textbf{Middle centre:} Argmax (worst–case) trajectories.
  \textbf{Top right:} \emph{Exact} decomposition at the minimiser: the physical term $\sum_i\!\int\!\|\dot u_i\|^2 dt$ varies smoothly; the unweighted model activity $\sum_i\!\int\!\|\dot m_i\|^2 dt$ peaks near $\kappa_c$; the \emph{weighted} contribution $\sum_i \lambda_i^{\*}(\kappa)\!\int\!\|\dot m_i\|^2 dt$ collapses post–transition. The thin line (right axis) shows the residual $E_{\min}-\big(E_{\mathrm{phys}}+\sum_i\lambda_i^{\*}E_{\mathrm{model},i}\big)\!\approx\!0$, confirming exact closure.
  \textbf{Bottom left:} Cross–sections $\lambda_1\mapsto \min_{\lambda_2}E$ at low $\kappa$, $\kappa_c$, and high $\kappa$ (with $\lambda_2$ fixed at its argmin), illustrating controller–specific non‑monotonicity.
  \textbf{Bottom centre:} Optimal path in $(\lambda_1,\lambda_2)$ coloured by $\kappa$; a star marks $\kappa_c$.
  \textbf{Bottom right:} Log–space argmin jump metric $\|\Delta\log\lambda^{\*}\|$; the peak identifies $\kappa_c$.
  All panels use fixed horizon $T$, identical initial conditions, and energies $E=\sum_i\int\|\dot u_i\|^2 dt+\lambda_i\int\|\dot m_i\|^2 dt$.}
  \label{fig:rtas_phase_transition}
\end{figure}

\paragraph{Simulation protocol.}
We integrate \eqref{eq:2agent-lift} with fixed time step.  Three classes of tasks specify $\dot c_i$: (i)~\emph{formation tracking}, where both agents follow a figure‑8 reference with lateral offset; (ii)~\emph{pursuit–evasion}, where the pursuer points toward the evader while the evader executes an escape field with a small oscillatory component; (iii)~\emph{diagnostics over $(\lambda_1,\lambda_2)$ and $\kappa$}, holding the target field simple (a unit circle) to scan for stability and resonances.  Metrics reported in the figures are total energy, formation error, model divergence $|m_1{-}m_2|$, stability (variance of $m_i$), and synchrony (corr$(m_1,m_2)$).

\paragraph{Policy-induced phase transition (summary).}
Before turning to the four experiments, we quantify a global optimisation effect of the coupling channel:
as $\kappa$ increases, the branchwise minima of the reflective cost cross, producing a first-order–like
switch between an adaptive low-$\lambda$ regime and a rigid high-$\lambda$ regime. The optimiser’s
$\lambda^{\ast}(\kappa)$ jumps at a critical $\kappa_c$, and RMS model activity collapses beyond $\kappa_c$.
See \Cref{fig:rtas_phase_transition} (panels a–d) for the energy landscape, the jump in $\lambda^\ast$,
model activity, and the resulting phase diagram.

\begin{figure}[h]
  \centering
  \includegraphics[width=0.8 \textwidth]{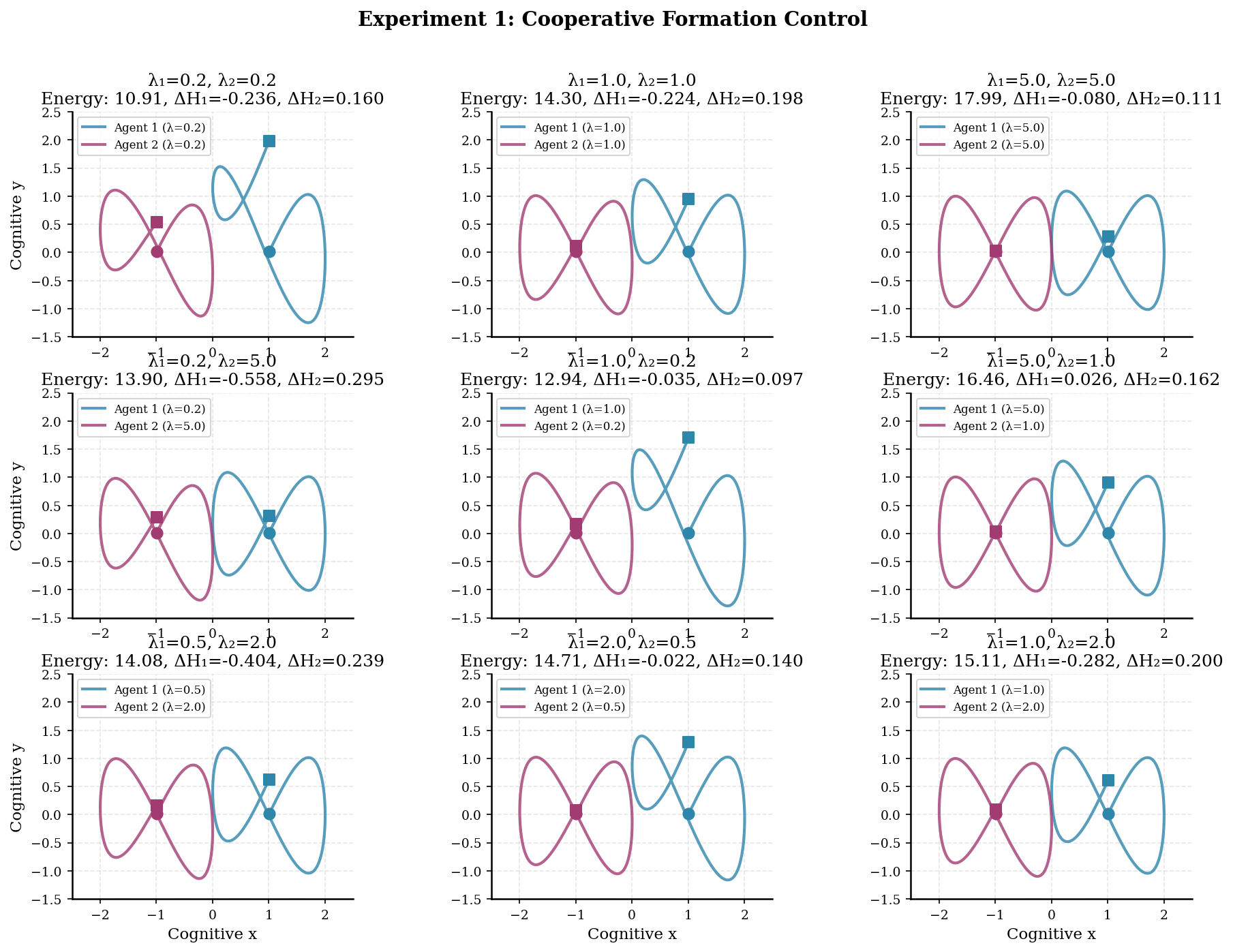}
\caption{\textbf{Experiment 1: Cooperative formation with co-adaptation.}
Two agents track a figure-8 in $C$ while maintaining a lateral offset.
Rows/columns sweep $(\lambda_1,\lambda_2)$; per-panel headers report total energy and a signed \emph{holonomy proxy} $\Delta H_i$, defined for the projective channel by
$\displaystyle \Delta H_i := \int_0^T F_{pm}\!\big(u_i(t)\big)\, m_i(t)\,dt$
with $F_{pp}=0$ and $F_{pm}\propto \cos u_i^{\,1}$.
\emph{Findings:} (i) Total energy generally increases with either $\lambda_i$, because model changes are penalised and effort shifts to the physical channel; (ii) When $(\lambda_1,\lambda_2)$ are mismatched, the low-$\lambda$ agent takes on more learning (larger $|m_i|$) and the high-$\lambda$ partner behaves as a rigid anchor (emergent leader follower split); (iii) Formation error remains small across conditions and is minimised when $\lambda_1\approx\lambda_2$ (symmetric responsibility for co-adaptation).
\emph{Note:} $\Delta H_i$ is a cross-curvature holonomy proxy and is \emph{not} the model change $\Delta m_i$. 
Agents have symmetric information coupling (both observe each other's model parameters).}

  \label{fig:ma_exp1}
\end{figure}

\paragraph{Cooperative formation (Fig.\,\ref{fig:ma_exp1}).}
With symmetric couplings $\kappa>0$ and identical targets, both agents solve a shared tracking task.  The coupling \eqref{eq:coupling} acts like a soft “perceptual spring”: each agent’s model $m_j$ alters the other’s effective $\dot c_i$, nudging the pair to negotiate a compromise between motion and learning.  Larger $\lambda$’s suppress $\dot m_i^\star$ and shift effort into $\dot u_i^\star$, raising total cost; mismatched $\lambda$’s yield asymmetric model excursions and a clear trade of learning load to the more adaptive partner.

\begin{figure}[h]
  \centering
  \includegraphics[width=\textwidth]{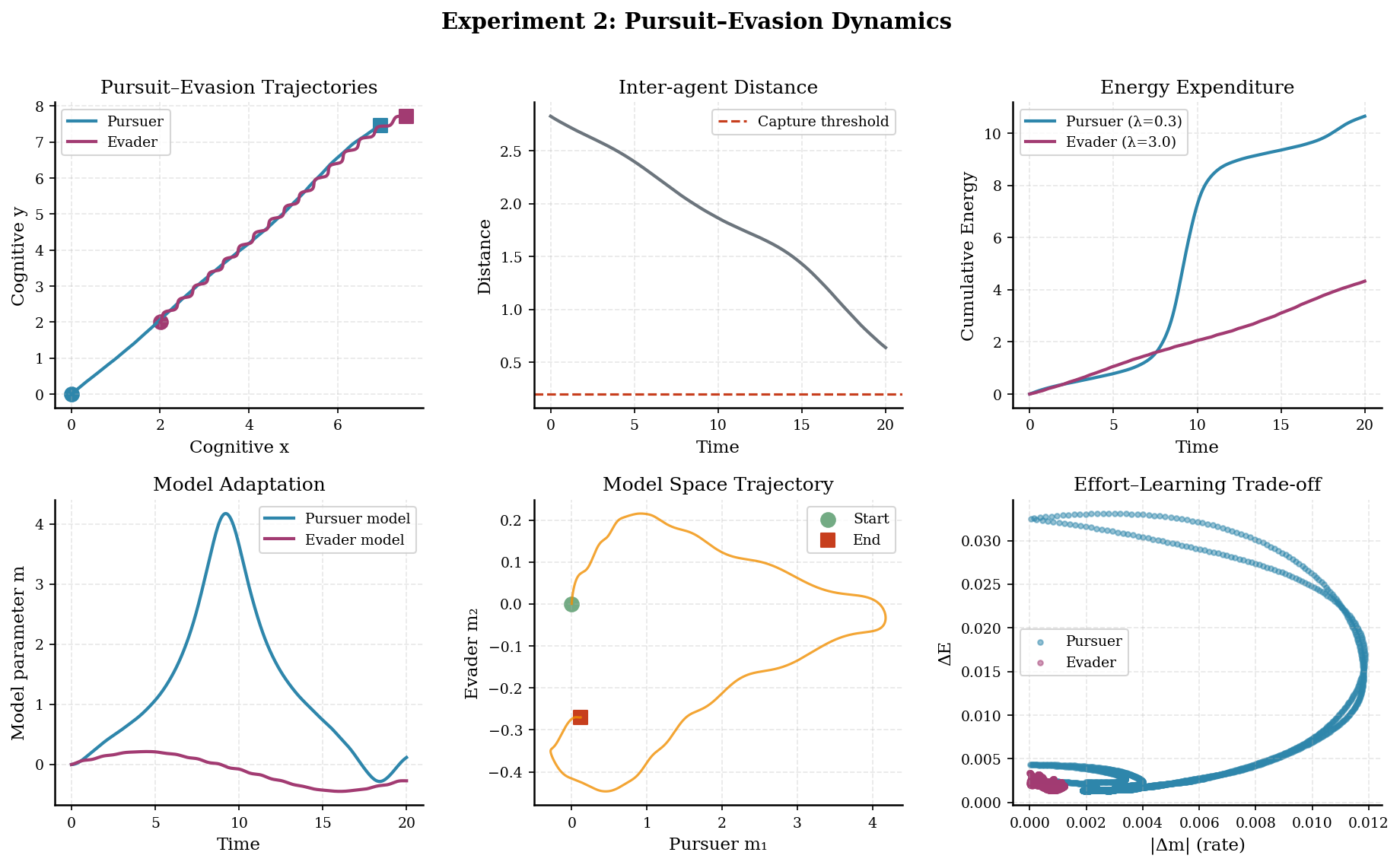}
  \caption{\textbf{Experiment 2 — Pursuit–evasion with information asymmetry.}
  The pursuer (low~$\lambda$) observes the evader’s model ($\kappa>0$ in~\eqref{eq:coupling}); the evader (high~$\lambda$) does not use the pursuer’s model (asymmetric coupling).  
  \emph{Findings:} the pursuer invests heavily in model updates (peaked $m_1$ and steep energy growth) to reshape its perceptual channel and reduce distance; the evader, being model‑rigid, exhibits small $m_2$ drift and lower energetic cost.  Distance decreases steadily but, under the chosen gains, does not cross the capture threshold—illustrating that reflective adaptation can improve interception geometry without guaranteeing capture when the opponent is sufficiently stiff.}
  \label{fig:ma_exp2}
\end{figure}

\paragraph{Pursuit–evasion with asymmetric information (Fig.\,\ref{fig:ma_exp2}).}
Setting \texttt{info\_asymmetry=True} implements a realistic asymmetry: the pursuer uses $m_2$ in~\eqref{eq:coupling} whereas the evader ignores $m_1$.  The pursuer’s low~$\lambda$ and access to the opponent’s model make it \emph{highly self‑referential}: it modifies its own $m_1$ to change what it perceives (via $B(u_1)$) and also to perturb what the evader would perceive if coupled.  The result is a large transient in $m_1$ (a “learning lunge”), increased energy, and a monotone reduction of inter‑agent distance.

\begin{figure}[h]
  \centering
  \includegraphics[width=\textwidth]{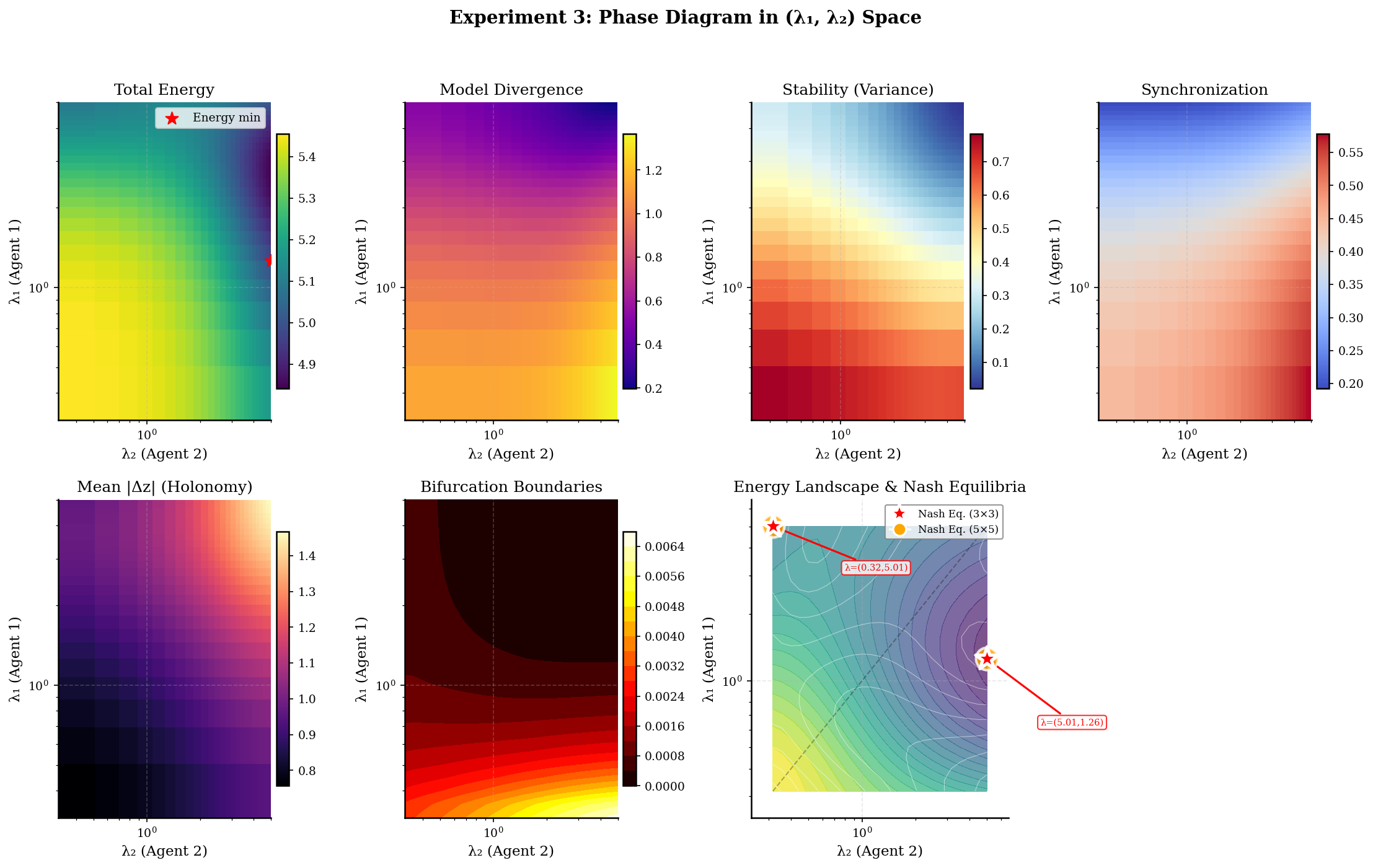}
\caption{\textbf{Experiment 3 — Phase diagram in $(\lambda_1,\lambda_2)$ space.}
Heat‑maps show (top row, left to right) total energy, average model divergence $\lvert m_1{-}m_2\rvert$, stability $\mathrm{Var}(m_i)$, and synchrony $\mathrm{corr}(m_1,m_2)$; (bottom row) mean $\lvert\Delta z\rvert$ (holonomy proxy), a finite‑difference bifurcation indicator, and (rightmost) the energy landscape with \emph{Nash‑like} points (red stars for $3{\times}3$ local searches; orange circles for $5{\times}5$ refinements) that serve as proxies for Nash equilibria (no best‑response maps are solved). 
\emph{Findings.} (i) For the \emph{policy used here}—the instantaneous reflective split of \eqref{eq:block_lift_formula}—the \emph{lowest sampled energy} occurs at an \emph{asymmetric} pair $(\lambda_1,\lambda_2)\!\approx\!(0.79,\,5.01)$, evidencing a leader–follower split. This plot therefore shows a \emph{policy‑induced} landscape rather than the globally re‑optimised value function; theoretically, the value function $\widehat{\mathcal E}^{\star}(\lambda_1,\lambda_2)$ is coordinatewise nondecreasing and concave (Proposition~\ref{prop:valfn-monotone}) and is thus minimised at small $(\lambda_1,\lambda_2)$. (ii) Strongly mismatched pairs maximise model divergence $\lvert m_1{-}m_2\rvert$, echoing the leader–follower split. (iii) Stability degrades as either agent becomes too flexible (small $\lambda_i$), while synchrony peaks near the diagonal $\lambda_1\!\approx\!\lambda_2$. (iv) Ridges in the bifurcation map delineate “safe’’ vs.\ “touchy’’ co‑adaptation regimes.}

  \label{fig:ma_exp3}
\end{figure}

\paragraph{Phase structure (Fig.\,\ref{fig:ma_exp3}).}
Scanning $(\lambda_1,\lambda_2)$ with fixed $\kappa$ reveals a simple design rule: \emph{match} the agents’ effort–learning exchange rates to prevent runaway responsibility transfer and to keep the pair in a high‑synchrony, low‑variance regime.  The local energy‑minimum map marks parameter pairs where neither agent could unilaterally lower total energy by changing \(\lambda_i\) (a coarse proxy for cooperative “Nash‑like” balance). This 2D scan complements the optimiser-switch behaviour summarised in \Cref{fig:rtas_phase_transition}.

\begin{figure}[h]
  \centering
  \includegraphics[width=\textwidth]{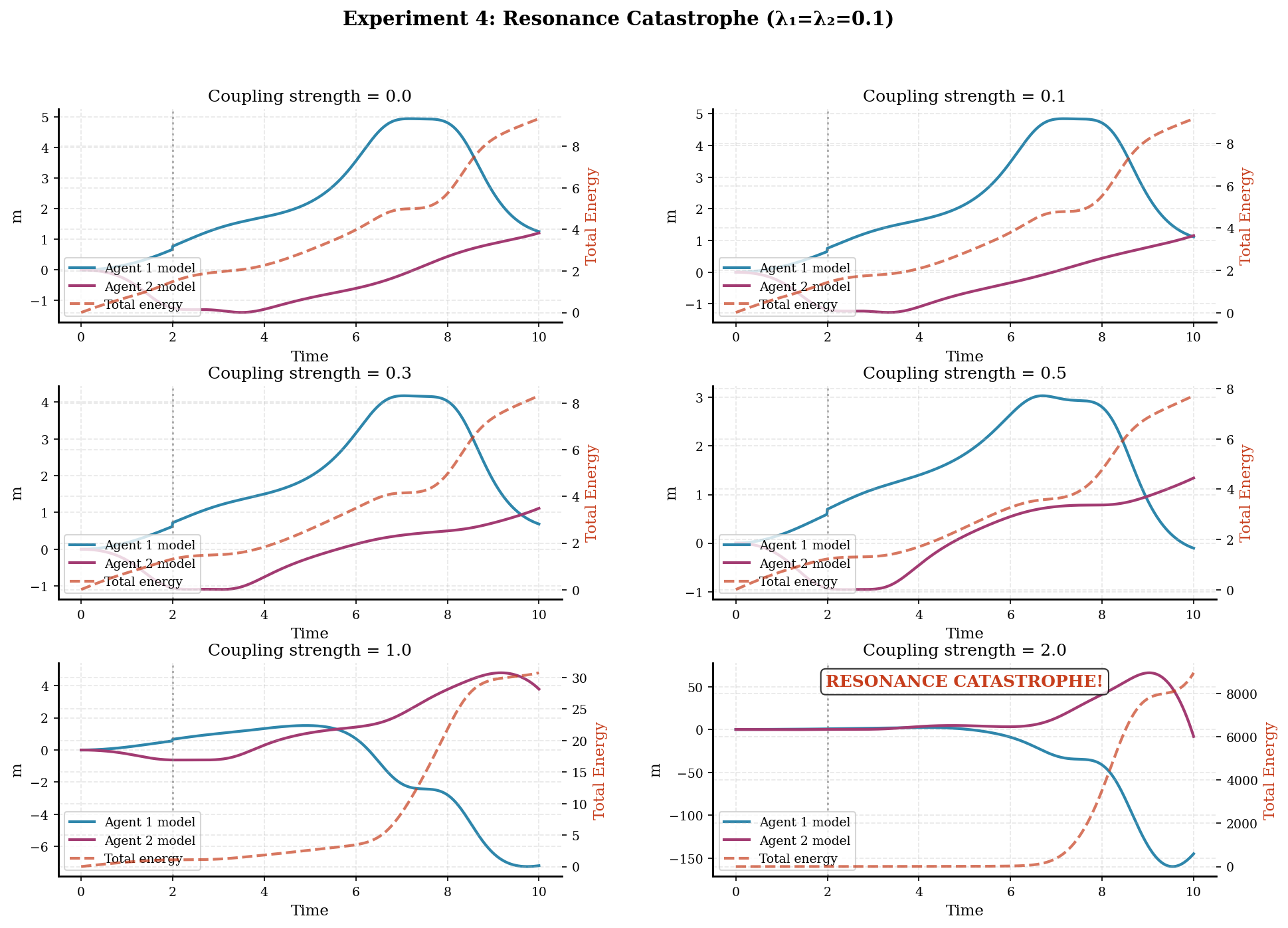}
  \caption{\textbf{Experiment 4 — Resonance catastrophe at low $\lambda$ and large $\kappa$.}
  With both agents highly adaptive ($\lambda_1{=}\lambda_2{=}0.1$) we sweep $\kappa$.  For small~$\kappa$, $m_i$ remain bounded; as $\kappa$ grows, cross‑feedback between \eqref{eq:coupling} and \eqref{eq:2agent-lift} creates a positive loop not damped by the $\lambda_i$, leading to explosive growth in $|m_i|$ and energy (lower‑right panel).
  \emph{Implication:} reflective coupling is powerful but must be regularised—either by larger $\lambda_i$, bounded $g(u)$, or saturation on $m_i$—to avoid self‑amplifying model dynamics.}
  \label{fig:ma_exp4}
\end{figure}

\paragraph{Take‑away.}
rTAS yields a compact, implementable recipe for multi‑agent, \emph{self‑referential} interaction: stack the per‑agent reflective lifts \eqref{eq:2agent-lift}; couple them through the effective cognitive targets \eqref{eq:coupling}; and price learning through $\lambda_i$.  The figures show three generic phenomena that persist across channels: (i)~\emph{effort–learning exchange} (who adapts pays the energy); (ii)~\emph{role emergence} under mismatched $\lambda_i$ (adaptive follower vs.\ rigid leader); and (iii)~\emph{resonant co‑adaptation} when high adaptability meets strong coupling.  These are exactly the multi‑agent counterparts of the single‑agent cost–memory law and curvature‑induced holonomy developed earlier. 

\newpage
\section{Discussion}

This work introduces Tangential Action Spaces (TAS) to formalize a foundational trade-off in agency: the energetic cost of memory. By representing agents as a hierarchy of manifolds ($P \to C \to I$), we have shown that the geometry of the perception-action map $\Phi$ dictates not only the energetic cost of an action but also whether that action leaves a path-dependent trace. Our central finding is a cost-memory duality: memory, in the form of holonomy, is never free. It requires deviating from the instantaneously energy-minimizing metric lift (\Cref{prop:metriclift}, \Cref{lem:optimality}), and this deviation incurs a quantifiable excess energetic cost (\Cref{thm:cost-memory}). For small loops, this excess cost scales quadratically with the magnitude of the stored memory (\Cref{prop:smallloop}).

\subsection{A Unified Geometric Framework for Cost, Memory, and Learning}

The primary contribution of this paper is a single, unified language for relating embodiment, cost, and cognition. The classification of systems into four archetypes—intrinsically conservative, conditionally conservative, geometrically nonconservative, and dynamically nonconservative—provides a principled taxonomy for understanding why some systems are history-dependent while others are not, as illustrated in our canonical examples (\Cref{fig:flat,fig:strip_sine,fig:helical,fig:twisted}).

Crucially, the framework extends beyond simple memory to encompass learning and self-modification through Reflective TAS (rTAS). By introducing a model manifold $M$, the parameter $\lambda$ in the block metric $\widehat{G} = G \oplus \lambda H$ prices model motion and allocates effort between physical action and model updates via the reflective metric lift (\Cref{def:reflective_lift}, \Cref{lem:instant_opt}). Cross-curvature terms ($F_{pm}, F_{mp}$) in the reflective connection provide a geometric mechanism for re-entry, where physical and model states reciprocally induce holonomy in one another (\Cref{fig:rtas2}). This makes the trade-off between physical effort and model change explicit and geometrically grounded (\Cref{fig:rtas1}, \Cref{fig:rtas3}).

\subsection{Dialogue with Existing Paradigms}

TAS provides a geometric foundation for several existing paradigms. It complements enactivist and dynamical systems accounts by identifying the specific mathematical source of history-dependence (curvature or prescribed dynamics). It refines efficiency-based models from optimal motor control by establishing the metric lift as the instantaneous energetic baseline, against which the cost of any memory-encoding deviation can be measured. Finally, it offers an implementation layer for predictive processing, suggesting that agents must balance the informational imperative to minimise surprise against the physical imperative to conserve energy.

\subsection{From Single Agents to Social Dynamics}

The multi-agent rTAS simulations demonstrate how local, energy-based rules can generate complex social phenomena (\Cref{fig:ma_exp1,fig:ma_exp2,fig:ma_exp3,fig:ma_exp4}). Under our chosen coupling architecture, we observe:
\begin{itemize}
    \item Emergent Role Specialization: Mismatched learning costs ($\lambda_1 \neq \lambda_2$) lead to spontaneous leader-follower dynamics (\Cref{fig:ma_exp1}, \Cref{fig:ma_exp3}).
    \item First-order like optimizer switch: Under the specified cost and coupling, the system can undergo a discontinuous change in its optimal strategy as coupling strength varies, switching between adaptive and rigid regimes (\Cref{fig:rtas_phase_transition}).
    \item Runaway Model Activity: For large coupling $\kappa$ and small learning costs $\lambda_i$, the system can exhibit resonant instabilities, leading to runaway model activity and energy expenditure (\Cref{fig:ma_exp4}).
\end{itemize}
These results illustrate that rTAS provides a bottom-up mechanism for grounding social dynamics in the individual energetic trade-offs of self-modifying agents.

\subsection{Testable Predictions and Design Principles}

The TAS framework suggests concrete, testable hypotheses for both biology and robotics.
\begin{itemize}
    \item For Biology: We predict that stereotyped motor tasks should correspond to policies approximating flat, energy-minimal metric lifts. In contrast, tasks requiring adaptation should engage pathways that produce holonomy at a measurable metabolic cost.
    \item For Robotics: The framework provides clear design principles. For high-efficiency tasks, robot morphologies and controllers should be co-designed to make the perception-action map $\Phi$ as "flat" as possible. Holonomy can be explicitly engineered for tasks requiring memory, with the cost-memory law (\Cref{prop:smallloop}) providing a budget for how much memory is sustainable.
\end{itemize}

\subsection{Operational Measurement}
The theory's predictions are empirically falsifiable. Given time-series data $(u(t), c(t))$, one can estimate $D\Phi$ by local linear regression. With a model for the physical metric $G$ (e.g., from actuator torque constants), one can form $M_C = D\Phi G^{-1} D\Phi^\top$ and compute the instantaneously optimal trajectory $\dot{u}^{\text{metric}} = G^{-1} D\Phi^\top M_C^{-1} \dot{c}$. The excess power is then $\|\dot{u}(t) - \dot{u}^{\text{metric}}(t)\|_G^2$. For closed cognitive loops, fibre holonomy is the physical closure gap $u(T) - u(0)$. One can then directly test the core predictions: whether $\|\Delta u_{\text{vert}}\|$ scales with enclosed area and whether the excess energy $\mathcal{E} - \mathcal{E}_{\text{metric}}$ scales with $\|\Delta u_{\text{vert}}\|^2$ for small loops.

\subsection{Limitations and Future Directions}

Our analysis rests on three standing assumptions: (A1) Abelian structure groups, (A2) time-invariant structures, and (A3) a quadratic effort functional. Relaxing these would be a key next step. Non-Abelian groups would require path-ordered integrals for holonomy, while time-varying bundles would introduce additional dynamic terms. Stochastic extensions are needed to model how noise competes with curvature-induced holonomy. The multi-agent results, in particular, open a new avenue for studying the co-evolution of embodiment and social strategy, where learning costs ($\lambda_i$) and coupling strengths ($\kappa$) become evolvable traits.

\section*{Acknowledgments}
The author thanks Michael Levin for valuable feedback on the biological implications of the framework; and Adam Goldstein for constructive comments on early drafts of this manuscript. Special thanks to Rena Seiler for discussions and inspiration on the overall idea of the TAS framework.
\section*{Data accessibility}
All data, code and scripts to reproduce the figures and analyses are available at
\href{https://github.com/marcelbtec/rTAS}{GitHub: rTAS} (tag \texttt{v1.0.0}).
\section*{Competing interests}
The author declares no competing interests.
\input{tas_appendix_proofs}
\bibliographystyle{unsrtnat}
\bibliography{tas_references}  

\end{document}

%% file: tas_appendix_proofs.tex

\newtheorem*{proposition*}{Proposition}
\newtheorem*{theorem*}{Theorem}

\appendix
\section{Detailed Proofs}
 Throughout, we adopt the following \emph{conventions}: (i) fibres are assumed Abelian (Assumption~\ref{ass:standing}); (ii) small‑loop statements are local (inside a single trivialisation); (iii) when comparing energies of lifts that depend on time parametrisation, we fix the loop parametrisation to have total time \(T=1\) unless stated otherwise.

\subsection{Large \texorpdfstring{$\lambda$}{lambda} behaviour of the reflective cost}
\label{subsec:lambda-monotonicity}

We record two facts that clarify what is (and is not) implied by increasing the
trade–off parameter $\lambda$.

\begin{proposition}[Value function is nondecreasing and concave in $\lambda$]
\label{prop:valfn-monotone}
Fix a $C^1$ cognitive loop $\gamma:[0,T]\!\to\!C$ and the admissible set
\[
\mathcal F \;:=\;\Big\{(u,m)\in H^1([0,T];P\times M)\ :\
D\Phi_m(p)\,\dot u+\partial_m\Phi_m(p)\,\dot m=\dot c\ \text{a.e.}\Big\},
\]
which does not depend on $\lambda$.
For $\lambda>0$ define the optimal reflective energy
\[
\widehat{\mathcal E}^\star(\lambda)
:=\inf_{(u,m)\in\mathcal F}\int_0^T\!\Big(\|\dot u(t)\|_G^2+\lambda\,\|\dot m(t)\|_H^2\Big)\,dt.
\]
Then $\lambda\mapsto \widehat{\mathcal E}^\star(\lambda)$ is nondecreasing and concave.
\end{proposition}

\begin{proof}
For each feasible $(u,m)\in\mathcal F$ define the affine map
\[
J_{(u,m)}(\lambda)
=\int_0^T\!\|\dot u\|_G^2\,dt \;+\; \lambda\!\int_0^T\!\|\dot m\|_H^2\,dt
=: a_{(u,m)} + b_{(u,m)}\,\lambda,
\quad b_{(u,m)}\ge 0.
\]
The value function is the pointwise infimum
$\widehat{\mathcal E}^\star(\lambda)=\inf_{(u,m)\in\mathcal F} J_{(u,m)}(\lambda)$.

\paragraph{Concavity.}
For any $\lambda,\mu>0$ and $t\in[0,1]$, linearity of $J_{(u,m)}$ gives
\[
J_{(u,m)}\!\big(t\lambda+(1-t)\mu\big)
=t\,J_{(u,m)}(\lambda)+(1-t)\,J_{(u,m)}(\mu).
\]
Taking the infimum over $(u,m)$ on both sides yields
\[
\widehat{\mathcal E}^\star\!\big(t\lambda+(1-t)\mu\big)
\;=\;\inf J_{(u,m)}\!\big(t\lambda+(1-t)\mu\big)
\;\ge\; t\,\widehat{\mathcal E}^\star(\lambda)+(1-t)\,\widehat{\mathcal E}^\star(\mu),
\]
hence $\widehat{\mathcal E}^\star$ is concave.

\paragraph{Monotonicity.}
If $\lambda_1\le\lambda_2$ then $J_{(u,m)}(\lambda_1)\le J_{(u,m)}(\lambda_2)$ for every
$(u,m)$ because $b_{(u,m)}\!\ge 0$. Taking infima preserves the inequality, so
$\widehat{\mathcal E}^\star(\lambda_1)\le \widehat{\mathcal E}^\star(\lambda_2)$.
\end{proof}
\begin{remark}
Concavity uses only that $\widehat{\mathcal E}^\star$ is the infimum of affine functions;
nonnegativity of the slopes $b_{(u,m)}$ is needed only for the monotonicity claim.
\end{remark}

\medskip
\noindent
Proposition~\ref{prop:valfn-monotone} states a general optimisation truth:
\emph{if one re-optimises globally as $\lambda$ changes, the optimal cost cannot
decrease.}  The decrease seen in \cref{fig:rtas1,fig:rtasA} therefore requires an
additional (benign) specialisation: those figures plot the cost of the
\emph{particular reflective split} produced pointwise by
\eqref{eq:block_lift_formula} along a fixed $\gamma$ (not the global-in-time
minimum over all feasible $(u,m)$).  For the projective channel this policy
admits a simple large–$\lambda$ bound that explains the observed
“decrease–then–saturate’’ trend.

\begin{proposition}[Projective channel: decreasing upper bound for large $\lambda$]
\label{prop:projective-upper-bound}
Consider the projective channel
$\Phi_m(u,v)=(u,\ v+m\sin u)$ with $G=\mathbf I_2$, $H=1$ and a fixed
$C^1$ loop $\gamma(t){=}(c_1(t),c_2(t))$.
Let $(u,v,m)$ be the trajectory obtained by the reflective split
\eqref{eq:block_lift_formula} (so $\dot u^\star=\dot c_1$ and
\(
\begin{bmatrix}\dot v^\star\\ \dot m^\star\end{bmatrix}
=\frac{1}{\lambda+\sin^2u}\begin{bmatrix}\lambda\\ \sin u\end{bmatrix}
(\dot c_2-m\cos u\,\dot c_1)
\)).
Assume $m(0)=0$ and set $E_{\mathrm{TAS}}:=\int_0^T(\dot c_1^2+\dot c_2^2)\,dt$.
Then there exists a constant $K\!=\!K(\|\dot c_1\|_{L^\infty},\|\dot c_2\|_{L^\infty},T)$
such that for all $\lambda\ge 1$
\begin{equation}
\widehat{\mathcal E}\big[(u,v,m)\big]
\;=\;\int_0^T\!\Big(\|\dot u^\star\|^2+\|\dot v^\star\|^2+\lambda\,|\dot m^\star|^2\Big)\,dt
\ \le\ E_{\mathrm{TAS}}\;+\;\frac{K}{\lambda}.
\label{eq:upper-bound}
\end{equation}
Consequently $\lambda\mapsto \widehat{\mathcal E}[(u,v,m)]$ is eventually
decreasing and $\widehat{\mathcal E}[(u,v,m)]\to E_{\mathrm{TAS}}$ as $\lambda\to\infty$.
\end{proposition}

\begin{proof}
Write $s(t):=\sin u(t)$ and
$w(t):=\dot c_2(t)-m(t)\cos u(t)\,\dot c_1(t)$.  From the explicit split,
$\dot m^\star=\frac{s}{\lambda+s^2}\,w$ and $\dot v^\star=\frac{\lambda}{\lambda+s^2}\,w$,
whence
\[
\|\dot v^\star\|^2+\lambda|\dot m^\star|^2
=\frac{\lambda}{\lambda+s^2}\,w^2
\ \le\ w^2.
\]
The $m$–dynamics is the linear ODE
$\dot m^\star+k_\lambda(t)\,\cos u(t)\,\dot c_1(t)\,m^\star=k_\lambda(t)\,\dot c_2(t)$
with $k_\lambda(t):=\frac{s(t)}{\lambda+s(t)^2}$.
Since $|k_\lambda|\le 1/\lambda$ and $\cos u\,\dot c_1$ is bounded, the
variation‑of‑constants formula and Grönwall give the uniform estimate
$\|m^\star\|_{L^\infty}\le C/\lambda$ for some
$C=C(\|\dot c_1\|_{L^\infty},\|\dot c_2\|_{L^\infty},T)$.
Hence $w=\dot c_2+O(\lambda^{-1})$ and
\[
\int_0^T\!\Big(\|\dot v^\star\|^2+\lambda|\dot m^\star|^2\Big)\,dt
\ \le\ \int_0^T\!\dot c_2^2\,dt\ +\ O(\lambda^{-1}).
\]
Adding $\int_0^T\dot c_1^2\,dt$ yields \eqref{eq:upper-bound}.
\end{proof}

\paragraph{Interpretation.}
Proposition~\ref{prop:projective-upper-bound} shows that, for the reflective
split used in the figures, penalising model motion suppresses $m$ as
$\|m\|_\infty=O(\lambda^{-1})$, which collapses the mixing term
$w=\dot c_2-m\cos u\,\dot c_1$.  The model‑activity cost scales as
$\lambda\!\int\!|\dot m^\star|^2\,dt=O(\lambda^{-1})$ and the physical part tends
to the TAS baseline $E_{\mathrm{TAS}}$.  Thus the measured reflective cost along
this policy decreases with~$\lambda$ and saturates, exactly as seen in
\cref{fig:rtas1}\,(c) and \cref{fig:rtasA}\,(d).  This behaviour is compatible
with Proposition~\ref{prop:valfn-monotone}: the globally optimised value
$\widehat{\mathcal E}^\star(\lambda)$ is nondecreasing in $\lambda$, while the
cost of the particular (instantaneous) reflective split enjoys the
large–$\lambda$ decreasing upper bound \eqref{eq:upper-bound}.


\subsection{Proof of Pythagorean Decomposition (Remark~\ref{rem:pythagorean})}

\begin{lemma}[Pythagorean decomposition (metric lift)]
\label{lem:pythagorean_full}
Let $\Phi: P \to C$ be a submersion and let $G$ be a Riemannian metric on $P$. For any $\Delta c \in T_{\Phi(p)}C$, let
\[
\Delta u^{\text{metric}} = G^{-1}D\Phi_p^\top \left(D\Phi_p G^{-1} D\Phi_p^\top\right)^{-1} \Delta c
\]
be the metric lift. Then for any $\Delta u \in T_pP$ satisfying $D\Phi_p(\Delta u) = \Delta c$:
\begin{enumerate}
    \item[(i)] There exists a unique $v \in \ker D\Phi_p$ such that $\Delta u = \Delta u^{\text{metric}} + v$.
    \item[(ii)] The decomposition is $G$-orthogonal: $\langle \Delta u^{\text{metric}}, v \rangle_G = 0$.
    \item[(iii)] The norms satisfy: $\|\Delta u\|_G^2 = \|\Delta u^{\text{metric}}\|_G^2 + \|v\|_G^2$.
\end{enumerate}
\end{lemma}

\begin{proof}
Part (i) is immediate by setting $v:=\Delta u-\Delta u^{\text{metric}}$ and noting $D\Phi_p v=0$.
For (ii),
\[
\langle \Delta u^{\text{metric}}, v \rangle_G
= (G^{-1}D\Phi^\top M^{-1}\Delta c)^\top G\,v
= \Delta c^\top M^{-1} D\Phi\,v
= 0,\qquad M:=D\Phi G^{-1}D\Phi^\top,
\]
because $v\in\ker D\Phi$. The identity in (iii) follows from the polarization identity and (ii).
\end{proof}


\subsection{Small–loop law (Proposition~\ref{prop:smallloop}) — Corrected, relative version}

\begin{proposition*}[Small–loop law (relative to the metric lift), restated]
Let $\Phi:P\to C$ be a fibration, $G$ a Riemannian metric on $P$, and let $L_{\mathrm{met}}$ denote the metric lift. Let $\nabla$ be any Ehresmann connection with horizontal map $L_\nabla$. For a $C^1$ closed loop $\gamma:[0,1]\to C$ contained in a trivialising chart, define 
\[
v(t):=L_\nabla(\dot c(t))-L_{\mathrm{met}}(\dot c(t))\in\ker D\Phi_{p(t)}.
\]
Then
\begin{equation}
\label{eq:relative_energy_identity}
\mathcal E_\nabla[\gamma]-\mathcal E_{\mathrm{met}}[\gamma]
= \int_0^1 \|v(t)\|_G^2\,dt
\ \ge\ \Big\|\int_0^1 v(t)\,dt\Big\|_G^2
= \big\|\Delta u_{\nabla}-\Delta u_{\mathrm{met}}\big\|_G^2.
\end{equation}
If, in addition, fibres are Abelian and the loop is sufficiently small, there is a constant 
$K$ (depending on a neighborhood of $\gamma$) such that
\begin{equation}
\label{eq:relative_area_bound}
\big\|\Delta u_{\nabla}-\Delta u_{\mathrm{met}}\big\|_G \ \le\ 
\|F_\nabla - F_{\mathrm{met}}\|_{L^\infty}\,\mathrm{Area}(\gamma),
\end{equation}
where $F_\nabla$ and $F_{\mathrm{met}}$ are the curvature $2$‑forms of the two connections in the chosen trivialisation.
\end{proposition*}

\begin{proof}
By Lemma~\ref{lem:pythagorean_full} applied pointwise with $\Delta c=\dot c(t)$ and $\Delta u=L_\nabla(\dot c(t))$,
\[
\|L_\nabla(\dot c)\|_G^2=\|L_{\mathrm{met}}(\dot c)\|_G^2+\|v\|_G^2,
\]
which integrates to the identity in \eqref{eq:relative_energy_identity}. The inequality in \eqref{eq:relative_energy_identity} is Cauchy–Schwarz in $L^2([0,1];V)$:
\(
\int_0^1\|v\|_G^2\,dt\ge \big\|\int_0^1 v\,dt\big\|_G^2.
\)
For \eqref{eq:relative_area_bound}, assume Abelian fibres and work in a local trivialisation with connection one‑forms $\omega_\nabla,\omega_{\mathrm{met}}$. Let $\Delta\omega:=\omega_\nabla-\omega_{\mathrm{met}}$ and $\Delta F:=d\Delta\omega=F_\nabla-F_{\mathrm{met}}$. Stokes' theorem gives
\[
\Delta u_{\nabla}-\Delta u_{\mathrm{met}}
=\oint_\gamma\Delta\omega
=\iint_{S_\gamma}\Delta F,
\]
whence the stated bound with the $L^\infty$ norm of $\Delta F$ over a neighborhood containing $S_\gamma$.
\end{proof}

\begin{remark}[On absolute holonomy and baselines]
If the metric connection itself has nonzero curvature, then \(\Delta u_{\mathrm{met}}\) can be \(O(\mathrm{Area})\), i.e.\ geometric holonomy may occur \emph{at baseline cost}. The proposition therefore quantifies the excess energy and holonomy \emph{relative} to the metric lift. If, in a given model, the metric connection is flat near the loop, then the same formulas control the absolute holonomy.
\end{remark}


\subsection{Block Pythagorean decomposition for rTAS}

\begin{lemma}[Block weighted pseudoinverse and orthogonality]
\label{lem:block_pythagoras}
Let $A\in\mathbb R^{n\times (p+q)}$ have full row rank, let $W=\operatorname{diag}(G,\lambda H)$ be a positive‑definite block weight with $G\in\mathbb R^{p\times p}$, $H\in\mathbb R^{q\times q}$, $\lambda>0$, and let $b\in\mathbb R^n$. The $W$‑least‑norm solution to $A y=b$ is
\[
y^\star=W^{-1}A^\top\big(AW^{-1}A^\top\big)^{-1}b.
\]
For any feasible $y$ with $Ay=b$ we have the orthogonal decomposition
\[
y=y^\star+z,\qquad z\in\ker A,\qquad \langle y^\star,z\rangle_W=0,
\]
and the $W$‑Pythagorean identity
\[
\|y\|_W^2=\|y^\star\|_W^2+\|z\|_W^2.
\]
\end{lemma}

\begin{proof}
Lagrange multipliers yield the stated $y^\star$. Then $W y^\star=A^\top\mu$ for some $\mu$, hence
\(
\langle y^\star,z\rangle_W=(A^\top\mu)^\top z=\mu^\top(Az)=0
\)
for $z\in\ker A$. The norm identity follows.
\end{proof}

Applied pointwise to $A=[\,D\Phi_m\ \ \partial_m\Phi_m\,]$ and $W=\widehat G=\operatorname{diag}(G,\lambda H)$, Lemma~\ref{lem:block_pythagoras} gives the instantaneous orthogonal split used below.


\subsection{Proof of Theorem~\ref{thm:extended_tradeoff} (Extended cost–memory duality)}

\begin{theorem*}[Extended cost–memory duality, restated (local, small‑loop)]
Consider a reflective TAS with block metric $\widehat{G}=G\oplus\lambda H$ and reflective constraint
$D\Phi_m(p)\,\dot u+\partial_m\Phi_m(p)\,\dot m=\dot c$. For a closed $C^1$ loop $\gamma:[0,1]\to C$ contained in a trivialisation, let $\widehat\gamma(t)=(u(t),m(t))$ be any reflective lift and let $(\dot u^\star,\dot m^\star)$ denote the instantaneous reflective metric lift \eqref{eq:block_lift_formula}. Then:
\begin{enumerate}[label=(\roman*)]
    \item The reflective metric lift minimises $\widehat{\mathcal E}[\widehat\gamma]=\int_0^1(\|\dot u\|_G^2+\lambda\|\dot m\|_H^2)\,dt$ among all admissible reflective lifts.
    \item (Small‑loop quadratic law, relative version) Let $\Delta u_{\text{phys}}$ and $\Delta m_{\text{model}}$ be the total holonomies of $\widehat\gamma$, and $(\Delta u^\star_{\text{phys}},\Delta m^\star_{\text{model}})$ those of the reflective metric lift. There exists a positive‑definite quadratic form $Q$ (depending smoothly on the basepoint and $\widehat G$) such that
    \[
    \widehat{\mathcal{E}}[\widehat{\gamma}] - \widehat{\mathcal{E}}_{\min}[\gamma] \ \geq\ 
    \begin{bmatrix} \Delta u_{\text{phys}} - \Delta u^\star_{\text{phys}} \\[2pt] \Delta m_{\text{model}} - \Delta m^\star_{\text{model}} \end{bmatrix}^{\!\top}
    Q\,
    \begin{bmatrix} \Delta u_{\text{phys}} - \Delta u^\star_{\text{phys}} \\[2pt] \Delta m_{\text{model}} - \Delta m^\star_{\text{model}} \end{bmatrix}
    \ +\ o\!\big(\mathrm{Area}(\gamma)^2\big).
    \]
    In particular, the excess reflective energy is quadratic in the combined (physical, model) holonomy, up to higher‑order terms in loop area.
\end{enumerate}
\end{theorem*}

\begin{proof}
(i) Optimality follows from Lemma~\ref{lem:block_pythagoras} applied pointwise and integrating: for any admissible $(\dot u,\dot m)$,
\[
\|\dot u\|_G^2+\lambda\|\dot m\|_H^2
=\|\dot u^\star\|_G^2+\lambda\|\dot m^\star\|_H^2+\|\dot u-\dot u^\star\|_G^2+\lambda\|\dot m-\dot m^\star\|_H^2,
\]
with equality iff $(\dot u,\dot m)=(\dot u^\star,\dot m^\star)$ a.e.

(ii) Set $v_u:=\dot u-\dot u^\star$ and $v_m:=\dot m-\dot m^\star$. By Lemma~\ref{lem:block_pythagoras}, $[v_u;v_m]\in\ker A$ pointwise, hence
\[
\widehat{\mathcal E}[\widehat\gamma]-\widehat{\mathcal E}_{\min}[\gamma]
=\int_0^1\!\big(\|v_u\|_G^2+\lambda\|v_m\|_H^2\big)\,dt
\ \ge\ 
\Big\|\int_0^1 v_u\,dt\Big\|_G^2+\lambda\Big\|\int_0^1 v_m\,dt\Big\|_H^2,
\]
by Cauchy–Schwarz in $L^2([0,1])$ and $T=1$. The right‑hand side is the squared $\widehat G$–norm of the \emph{holonomy difference} vector
\[
\begin{bmatrix}\Delta u_{\text{phys}}-\Delta u^\star_{\text{phys}}\\ \Delta m_{\text{model}}-\Delta m^\star_{\text{model}}\end{bmatrix}.
\]
For small loops in a trivialisation (Abelian case), each holonomy component is given by a surface integral of the corresponding curvature block (and cross‑curvature) plus $o(\mathrm{Area})$ terms; thus the holonomy difference is $O(\mathrm{Area})$, and the quadratic form
\(
Q=\widehat G/T + O(\mathrm{Area})
\)
(transported to the basepoint) is positive definite for sufficiently small loops. This yields the claim with the stated $o(\mathrm{Area}^2)$ remainder.
\end{proof}

\begin{remark}[Absolute versus relative form]
If the reflective metric connection yields zero holonomy to first order (e.g.\ it is flat in the chosen channel near the loop), then the right‑hand side simplifies to a quadratic lower bound in the \emph{absolute} holonomy vector $\big(\Delta u_{\text{phys}},\Delta m_{\text{model}}\big)$.
\end{remark}


\subsection{Supporting Lemmas}

\begin{lemma}[Invertibility of the moment matrix]
\label{lem:moment_invertible}
Under the conditions of Proposition~\ref{prop:metriclift}, the matrix $M = D\Phi_p G^{-1} D\Phi_p^\top$ is symmetric positive definite and hence invertible.
\end{lemma}

\begin{proof}
Symmetry is clear. For positive definiteness, let $y \in \mathbb{R}^n\setminus\{0\}$. Then
\[
y^\top M y = y^\top D\Phi_p G^{-1} D\Phi_p^\top y = \|D\Phi_p^\top y\|_{G^{-1}}^2>0,
\]
because $D\Phi_p$ has full row rank (hence $D\Phi_p^\top y\neq 0$) and $G^{-1}>0$.
\end{proof}

\begin{lemma}[Continuity of (Abelian) holonomy]
\label{lem:holonomy_continuous}
For a smooth family of loops $\gamma_s$ varying continuously in the $C^1$ topology, the (Abelian) holonomy $\Delta u_{\text{vert}}(s)=\iint_{S_{\gamma_s}}F$ varies continuously in $s$.
\end{lemma}

\begin{proof}
In a trivialisation, holonomy is the surface integral of the smooth curvature $2$‑form $F$ over $S_{\gamma_s}$. Continuous dependence on $s$ follows from dominated convergence.
\end{proof}